\documentclass[a4paper,11pt]{article}

\usepackage{amsmath,amssymb,amsthm}
\usepackage{fancyvrb}
\usepackage{graphicx}
\usepackage{url}

\usepackage[margin=18pt,labelsep=period,font={normal,up},labelfont={bf,up}]{caption}

\usepackage{subfig}

\renewcommand{\geq}{\geqslant}
\renewcommand{\leq}{\leqslant}
\renewcommand{\epsilon}{\varepsilon}

\newcommand{\calB}{\ensuremath{\mathcal{B}}}
\newcommand{\calC}{\ensuremath{\mathcal{C}}}
\newcommand{\calE}{\ensuremath{\mathcal{E}}}

\newcommand{\calG}{\ensuremath{\mathcal{G}}}
\newcommand{\calM}{\ensuremath{\mathcal{M}}}

\newcommand{\calT}{\ensuremath{\mathcal{T}}}
\newcommand{\calL}{\ensuremath{\mathcal{L}}}

\newcommand{\numP}{\ensuremath{\#\mathbf{P}}}

\newcommand{\ball}{\ensuremath{\textrm{Ball}}}
\newcommand{\dtv}{\ensuremath{d_\textrm{tv}}}
\newcommand{\prob}{\ensuremath{\mathrm{Pr}}}
\newcommand{\expect}[1]{\ensuremath{\mathbb{E}[#1]}}
\newcommand{\Bprime}{\ensuremath{{B'\!}}}
\newcommand{\calBprime}{\ensuremath{{\mathcal{B}'\!}}}

\newcommand{\showgraphics}[1]{\includegraphics[scale=0.9]{figures/#1}}
\newcommand{\trilat}{\ensuremath{\mathcal{T}}}

\newcommand{\verbatimStart}[1]{\begin{Verbatim}[frame=lines,gobble=0,numbersep=0.3cm,numbers=left,xleftmargin=0.7cm,numberblanklines=false,firstnumber=#1]}

\theoremstyle{plain}
\newtheorem{theorem}{Theorem}[]
\newtheorem{lemma}[theorem]{Lemma}
\newtheorem{corollary}[theorem]{Corollary}

\newtheorem{definition}[theorem]{Definition}

\theoremstyle{definition}

\newtheorem*{proviso*}{Proviso}
\newtheorem*{note*}{Note}

\newcounter{tempcounter}

\title{Sampling Colourings of the Triangular Lattice}
\author{Markus Jalsenius\\~\\Department of Computer Science\\University of Bristol\\Bristol, BS8 1UB, UK}
\date{}

\begin{document}
\maketitle

\begin{abstract}
We show that the Glauber dynamics on proper 9-colourings of the triangular lattice is rapidly mixing, which allows for efficient sampling. Consequently, there is a fully polynomial randomised approximation scheme (FPRAS) for counting proper 9-colourings of the triangular lattice. Proper colourings correspond to configurations in the zero-temperature anti-ferro\-magnetic Potts model. We show that the spin system consisting of proper 9-colourings of the triangular lattice has strong spatial mixing. This implies that there is a unique infinite-volume Gibbs distribution, which is an important property studied in statistical physics. Our results build on previous work by Goldberg, Martin and Paterson, who showed similar results for 10~colours on the triangular lattice. Their work was preceded by Salas and Sokal's 11-colour result. Both proofs rely on computational assistance, and so does our 9-colour proof. We have used a randomised heuristic to guide us towards rigourous results.
\end{abstract}


\section{Introduction}
\label{sec:introduction}

This paper is concerned with proper 9-colourings of the triangular lattice. A \emph{$q$-colouring} of a graph $G$ is an assignment of colours from the set $\{1,\dots,q\}$ to the vertices. A colouring is \emph{proper} if adjacent vertices receive different colours. There are two fundamental problems that have been studied extensively: sampling $q$-colourings from the uniform distribution on proper $q$-colourings of $G$ and counting the number of proper $q$-colourings of $G$. Both these problems have been studied for various graphs and number of colours, and it has been shown that the two problems are intimately related. More precisely, if there is an efficient method of finding a good approximate solution to one of the problems, then there is an efficient method of finding a good approximate solution to the other problem. See for instance~\cite{j-csi-03} and~\cite{jvv-rgcs-86} for details on this topic. The problem of counting $q$-colourings of a graph with maximum degree $\Delta$ is \numP-complete when $q\geq 3$ and $\Delta\geq 3$ (see~\cite{bdgj-accsdg-99}). Hence we have to rely on approximate solutions. A great deal of work has been put into mapping out for which graphs and number of colours there is a fully polynomial randomised approximation scheme (FPRAS) for counting colourings, or equivalently, a fully polynomial almost uniform sampler (FPAUS).

Sampling colourings is closely related to problems arising in statistical physics, where proper colourings correspond to configurations in the \emph{zero-temperature anti-ferromagnetic Potts model}. The notion of a proper colouring imposes a very \emph{local} property: it constrains the allowed colours on two adjacent vertices. Given local constraints, physicists are interested in understanding the \emph{macroscopic} properties of the system.

Let $\calL$ be an infinite graph. Think of $\calL$ as lattice graph, for example the grid. For a finite subgraph $G$ of $\calL$, the \emph{boundary} of $G$ is the set of vertices in $\calL$ that are adjacent to $G$ but are not in $G$ themselves. Given a colouring $\calB$ of the boundary, a proper $q$-colouring of $G$ \emph{agrees} with $\calB$ if no vertex adjacent to the boundary receives any of the colours of its boundary neighbours. Let $\pi_\calB$ denote the uniform distribution on proper $q$-colourings of $G$ that agree with $\calB$. Suppose $\pi_\calL$ is a distribution on the set of proper $q$-colourings of the infinite graph $\calL$. For a proper $q$-colouring $\calC$ of $\calL$, let $\pi_\calL(\,\cdot \mid \calC(\calL\setminus G))$ denote the conditional distribution on proper $q$-colouring of $G$ induced by $\pi_\calL$ when the colours of all vertices except for those in $G$ are specified by $\calC$. The distribution $\pi_\calL$ is an \emph{infinite-volume Gibbs distribution} (with respect to proper colourings) if, for any proper $q$-colouring $\calC$ of $\calL$, $\pi_\calL(\,\cdot \mid \calC(\calL\setminus G))=\pi_\calB$, where $\calB$ is the colouring of the boundary of $G$ induced by $\calC$. That is, the conditional distribution on colourings of $G$ depends only on the boundary of $G$ and not on other vertices of $\calL$. The physical intuition for an infinite-volume Gibbs distribution is that it describes a macroscopic equilibrium, for which all parts of the system are in equilibrium with their boundaries. It is known that there always exists at least one infinite-volume Gibbs distribution~\cite{ghm-rgef-01}, and a central question in statistical physics is to determine whether it is unique or not. The phenomenon of non-uniqueness is referred to as a \emph{phase transition}. For more on Gibbs distributions, see, for example,~\cite{g-gmpt-88} or~\cite{ghm-rgef-01}.

Given a finite subgraph $G$ of $\calL$ and a boundary colouring $\calB$ of $G$, let $G'$ be any connected subset of the vertices of $G$ and let $\theta$ be the distribution on colourings of $G'$ that is induced by $\pi_\calB$. The question we ask is how does $\theta$ change if we change the colour of one single vertex on the boundary of $G$. If $G'$ is far away from the boundary then we would expect that changing the colour of a boundary vertex does not have too much influence on $\theta$. If this is true, then we have a property known as \emph{strong spatial mixing}. The exact definition will be given later. A consequence of strong spatial mixing is that the infinite-volume Gibbs distribution is unique~\cite{dsvw-mtslss-04,w-mtsdss-04,w-ccugm-05}. In this paper we show that there is strong spatial mixing for 9 colours on the triangular lattice.

Another important problem in statistical physics is to determine how quickly a system converges to equilibrium. This could lead to insights in how the system returns to equilibrium after a shock has disturbed it. A dynamical process that is commonly studied is the \emph{Glauber dynamics}. The Glauber dynamics is a Markov chain on the set of proper $q$-colourings of a graph $G$. A transition from one state (colouring) to another is made as follows. Choose a vertex $v\in G$ uniformly at random and choose one of the $q$ colours uniformly at random. If the randomly selected colour is different from the colours of the neighbours of $v$ then recolour $v$ with the selected colour. This procedure is also known as a \emph{single-vertex heat-bath update}. If it is repeated over and over then the distribution of the current state of the Markov chain will converge towards the uniform distribution. The question is how long to run the dynamics in order to get close to the uniform distribution. If we only need to simulate the dynamics over a small number of steps then the Glauber dynamics would be a suitable tool for sampling colourings. We say that a Markov chain is \emph{rapidly mixing} if a small number of steps are sufficient in order to get arbitrarily close to its stationary distribution. Exact definitions will be given later. It is a well-known fact that if the system has strong spatial mixing, then the Glauber dynamics is (often) rapidly mixing \cite{dsvw-mtslss-04,m-lgddsm-97,w-mtsdss-04}. In fact, the converse is also true. Since we prove strong spatial mixing for 9 colours on the triangular lattice, we also prove that the Glauber dynamics is rapidly mixing. That is, we show that there is an FPAUS for sampling 9-colourings of the triangular lattice. As mentioned above, the FPAUS implies that there is an FPRAS for counting 9-colourings.


\subsection{Organisation}

The remainder of the paper is organised as follows.
In Section~\ref{sec:preliminaries}, we define the basic notation, strong spatial mixing and the Glauber dynamics.
In Section~\ref{sec:related-work}, we formally state our results and discuss related work.
In Section~\ref{sec:boundary-pairs}, we introduce two key concepts that will be used throughout the paper, and Sections~\ref{sec:recursive} through~\ref{sec:exponential-decay} contain the technical part that leads to the strong spatial mixing result. Section~\ref{sec:glauber} is about the implication from strong spatial mixing to a rapidly mixing Glauber dynamics. Finally, in Section~\ref{sec:computations}, we describe the computer assisted parts of our proofs.


\section{Preliminaries}
\label{sec:preliminaries}

\subsection{Basic notation}

Let the infinite graph $\calT = (V_\calT, E_\calT)$ denote the \emph{triangular lattice}, formally defined as follows. There is a bijection $\xi$ from $V_\calT$ to $\{(x,y) \mid x,y \in\mathbb{Z}$ and $x + y$ is even\} such that $\xi^{-1}(x,y)\in V_\calT$ has the six neighbours $\xi^{-1}(x,y+2)$, $\xi^{-1}(x+1,y+1)$, $\xi^{-1}(x+1,y-1)$, $\xi^{-1}(x,y-2)$, $\xi^{-1}(x-1,y-1)$ and $\xi^{-1}(x-1,y+1)$. Note that $\xi$ maps a vertex of $V_\calT$ to a coordinate in a Cartesian coordinate system, which allows us to draw the triangular lattice as illustrated in Figure~\ref{fig:lattice}. However, throughout this article we will draw vertices as a hexagons, illustrated in Figure~\ref{fig:latticehex}.
\begin{figure}[t]
    \begin{minipage}[t]{0.47\linewidth}
        \centering
        \showgraphics{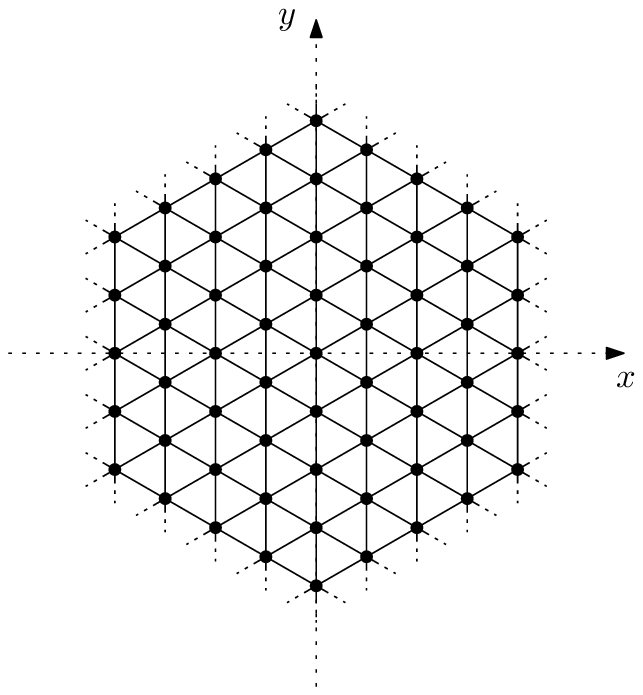}
        \caption{\label{fig:lattice}The triangular lattice.}
    \end{minipage}
    \hfill
    \begin{minipage}[t]{0.47\linewidth}
        \centering
        \showgraphics{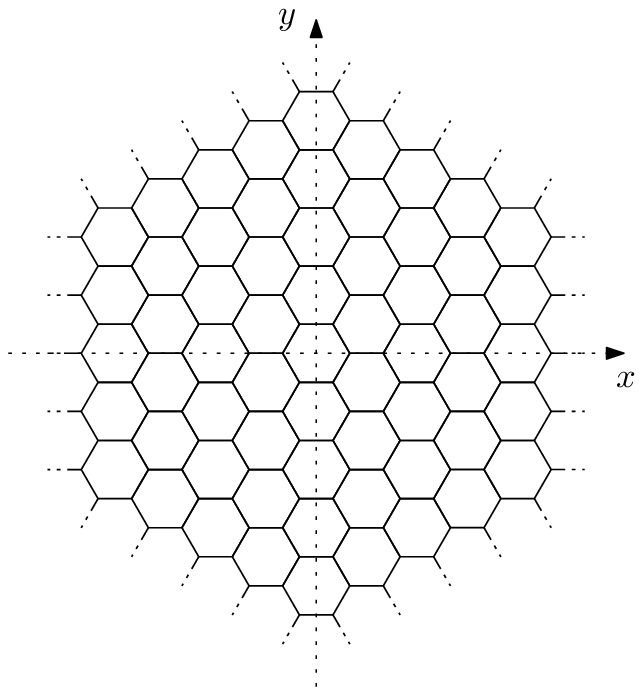}
        \caption{\label{fig:latticehex}The triangular lattice. Vertices are drawn as a hexagons.}
    \end{minipage}
\end{figure}

The bijection $\xi$ specifies a unique clockwise ordering of edges incident to the same vertex. We write $e_1\prec \cdots \prec e_k$ to indicate that $e_1 \rightarrow e_2 \rightarrow \cdots \rightarrow e_k \rightarrow e_1 \rightarrow \cdots$ is the clockwise ordering of the edges $e_1,\dots,e_k$ around some vertex $v$. We say that two edges $e_1$ and $e_2$ incident to the same vertex $v$ are \emph{clockwise adjacent} if either $e_1 \prec e_2 \prec e_3 \prec e_4 \prec e_5 \prec e_6$ or $e_2 \prec e_1 \prec e_3 \prec e_4 \prec e_5 \prec e_6$, where $e_3,\dots,e_6$ are the other four edges in $E_\calT$ incident to $v$.

A \emph{region} $R$ is a finite subset of $V_\calT$ and its \emph{edge boundary}, denoted $\calE R = \{\{v,w\} \mid \{v,w\}\in E_\calT, v\in R, w \notin R\}$.
The \emph{vertex boundary} of $R$, denoted $\partial R = \{w \mid \{v,w\}\in \calE R, w\notin R\}$.

For a region $R$, $v\in R$ and $w\in \partial R$, we let $d_R(w,v)$ denote the number of edges on a shortest path $P$ in $\calT$ between $w$ and $v$ such that all vertices in $P$ except for $w$ are in $R$. Thus, $d_R(w,v)=1$ if $v$ and $w$ are adjacent. For a subregion $R'\subseteq R$, we let $d_R(w,R')= \min_{u\in R'} d_R(w,u)$.

Let $[q]=\{1,\dots,q\}$. A \emph{$q$-colouring} $\calC$ of a set $S$ (a set of vertices or edges), is a function from $S$ to $[q]$, where $[q]$ represents a set of $q$ \emph{colours}. A \emph{partial $q$-colouring} of $S$ is a function from $S$ to $[q]\cup \{0\}$, where the additional colour~0 can be thought of as representing ``no colour.''

A $q$-colouring $\calC$ of a region $R$ is \emph{proper} if for all $u,v\in R$, $\calC(u)\neq \calC(v)$ when $\{u,v\}\in E_\trilat$. For a partial $q$-colouring $\calB$ of the boundary $\partial R$, a $q$-colouring $\calC$ of $R$ \emph{agrees} with $\calB$ if for all $\{w,v\}\in \calE R$, $\calB(w)\neq \calC(v)$. Similarly, for a partial $q$-colouring $B$ of $\calE R$, a $q$-colouring $\calC$ of $R$ agrees with $B$ if for all $\{w,v\}\in \calE R$, $B(\{w,v\})\neq \calC(v)$. To aid the reader, we
use the calligraphic font to represent colourings of vertices and the normal font to represent colourings of edges (like $\calB$ and $B$ above). Further, partial colourings (containing the ``colour''~0) are used only for \emph{boundary} vertices and \emph{boundary} edges and not for the regions themselves.

For a region $R$ and partial $q$-colouring $\calB$ of $\partial R$ (respectively, $B$ of $\calE R$),
$\Omega^q_\calB$ (respectively, $\Omega^q_B$) denotes the set of proper $q$-colourings of $R$ that agree with $\calB$ (respectively, $B$). The uniform distribution on $\Omega^q_\calB$ (respectively, $\Omega^q_B$) is denoted $\pi_\calB^q$ (respectively, $\pi_B^q$).
For a subregion $R'\subseteq R$, we let $\pi_\calB^q(R')$ denote the distribution on proper $q$-colourings of $R'$ induced by $\pi_\calB^q$.

For a distribution $\theta$ on a set $S$, we write $\prob_\theta(x)$ for the probability of drawing $x\in S$ under $\theta$. The \emph{total variation distance} between two distributions $\theta_1$ and $\theta_2$ on a set $S$ is
\begin{equation*}
    \dtv(\theta_1,\theta_2)=
        \frac{1}{2}\sum_{x\in S}
            |\prob_{\theta_1}(x)-\prob_{\theta_2}(x)| =
        \max_{S'\subseteq S}
            |\prob_{\theta_1}(S')-\prob_{\theta_2}(S')|\,.
\end{equation*}

For two distributions $\theta_1$ and $\theta_2$ on two sets $S_1$ and $S_2$, respectively, a \emph{coupling} $\Psi$ of $\theta_1$ and $\theta_2$ is a joint distribution on $S_1\times S_2$ that has $\theta_1$ and $\theta_2$ as its marginal distributions.

We write $\expect{X}$ to denote the expected value of a random variable $X$.


\subsection{Strong spatial mixing}
\label{sec:ssm}

We say that there is \emph{strong spatial mixing} for $q$ colours on the triangular lattice if there are two constants $\beta, \beta' > 0$ (that may depend on $q$) such that
\begin{equation}
    \label{eq:ssm}
    \dtv(\pi^q_{\calB}(R'),\pi^q_{\calB'}(R')) \leq
        \beta|R'|e^{-\beta'd_R(w,R')}
\end{equation}
is true for all regions $R$, subregions $R'\subseteq R$, $w\in \partial R$ and partial $q$-colourings $\calB$ and $\calB'$ of $\partial R$ such that $\calB(v)=\calB'(v)$ for all $v\in \partial R \setminus \{w\}$, $\calB(w)\neq \calB'(w)$, $\calB(w)> 0$ and $\calB'(w)> 0$.


\subsection{The Glauber dynamics}
\label{sub:Rapid-mixing}

The \emph{Glauber dynamics} for $q$-colourings of a region $R$ with partial $q$-colouring $\calB$ of $\partial R$ is a Markov chain with state space $\Omega^q_\calB$ and the following transitions. The evolution from a colouring $\calC_t\in \Omega^q_\calB$ to a new colouring $\calC_{t+1}\in \Omega^q_\calB$ is defined by the following steps.
\begin{enumerate}
    \setlength{\itemsep}{0pt}
    \item Choose a vertex $v$ from $R$ uniformly at random.

    \item[] Let $\calC$ be the colouring of $R\cup \partial R$ specified by $\calC_t$ and $\calB$.

    \item[] Let $\calC_\text{neighbours} = \{\calC(u) \mid \{u,v\}\in E_\trilat\}$.

    \item Choose a colour $c$ from $[q]\setminus \calC_\text{neighbours}$ uniformly at random.

    \item Let
        \begin{equation*}
            \calC_{t+1}(u)=
                \begin{cases}
                     c,          &\text{if $u = v$;}\\
                     \calC_t(u), &\text{otherwise.}
                \end{cases}
        \end{equation*}
\end{enumerate}

The probability of not leaving a state in a transition is always positive. Hence the Glauber dynamics is an ergodic Markov chain if the number of colours $q$ is sufficiently large; every state can be reached from any other state. In this article we are concerned with $q=9$ colours for which we note that there are always at least three available colours in Step~2 of the dynamics above. Hence the Glauber dynamics is ergodic for $q=9$.

It is straightforward to verify that when the Glauber dynamics is ergodic then its stationary distribution is $\pi_\calB^q$. Let $\theta^t_{\calC_0}$ be the distribution on $\Omega^q_\calB$ after $t$ steps of the Glauber dynamics, starting with colouring $\calC_0$. For $\delta>0$, the \emph{mixing time}
\begin{equation*}
    \tau_\calB^q(\delta) =
        \max_{\calC_0\in \Omega^q_\calB}
        \min_t\{t \mid
        \text{$\dtv(\theta^t_{\calC_0},\pi_\calB^q) \leq \delta\}$}
\end{equation*}
is the number of transitions until the dynamics is within total variation distance $\delta$ of the stationary distribution, assuming the worst initial colouring $\calC_0$. We say that the Glauber dynamics is \emph{rapidly mixing} if $\tau_\calB^q(\delta)$ is upper-bounded by a polynomial in $|R|$ and $\log(1/\delta)$.


\subsection{Approximate counting}

A \emph{randomised approximation scheme} (RAS) for a function $f:\Sigma^*\rightarrow\mathbb{N}$ is a probabilistic Turing machine that takes as input a pair $(x,\epsilon)\in \Sigma^*\times (0,1)$, and produces, on an output tape, an integer random variable~$Y$ satisfying the condition $\Pr(e^{-\epsilon} \leq Y/f(x) \leq e^\epsilon)\geq \frac{3}{4}$. The choice of the value $\frac{3}{4}$ is inconsequential: the same class of problems has a RAS if we choose any probability in the interval $(\frac{1}{2},1)$ (see for example~\cite{jvv-rgcs-86}). A \emph{fully polynomial randomised approximation scheme} (FPRAS) is a RAS that runs in time upper-bounded by a polynomial in $|x|$ and $\epsilon^{-1}$.

It was mentioned in the introduction that the existence of an efficient method for sampling colourings implies that there is an FPRAS for counting the number of colourings. We could use a rapidly mixing Glauber dynamics to construct (in a non-trivial way) an FPRAS for estimating $|\Omega^q_\calB|$. For details on the topic of how sampling and counting are related, see, for example, \cite{j-csi-03} or~\cite{jvv-rgcs-86}.


\section{Our results and related work}
\label{sec:related-work}

These are the two main theorems of the paper.

\newcommand{\ssmthm}{There is strong spatial mixing for 9 colours on the triangular lattice.}
\newcounter{ssmcounter}
\setcounter{ssmcounter}{\value{theorem}}
\begin{theorem}
    \label{thm:ssm}
    \ssmthm
\end{theorem}

\newcommand{\glauberthm}{The Glauber dynamics on 9-colourings of a region $R$ of the triangular lattice, with partial 9-colouring $\calB$ of $\partial R$, has mixing time $\tau_\calB^9(\delta)\in O(n^2+n\log\frac{1}{\delta})$, where $n=|R|$.}
\newcounter{glaubercounter}
\setcounter{glaubercounter}{\value{theorem}}
\begin{theorem}
    \label{thm:rapid-mixing}
    \glauberthm
\end{theorem}

The previously best known mixing results on the triangular lattice was given for 11~colours by Salas and Sokal~\cite{ss-apt-97} in~1997, and later improved by Goldberg, Martin and Paterson~\cite{gmp-ssmfc-04} to 10~colours in 2004. Both proofs involved computational assistance.

To place these results in context, we first mention some general mixing bounds that are applicable to graphs with small girth, such as many of those lattices studied in statistical physics. Independently, Jerrum~\cite{j-vsa-95} and Salas and Sokal~\cite{ss-apt-97} proved that for proper $q$-colourings on a graph of maximum degree $\Delta$, the Glauber dynamics has $O(n \log n)$ mixing time when $q > 2 \Delta$, where $n$ is the number of vertices. For $q = 2 \Delta$, Bubley and Dyer~\cite{bd-pctprmmc-97} showed that it mixes in $O(n^3)$ time, and later Molloy~\cite{m-vrmmc2dc-01} showed that it mixes in $O(n \log n)$ time. In~\cite{v-ibsc-00}, Vigoda used a Markov chain that differs from the Glauber dynamics and showed that it has $O(n \log n)$ mixing time when $q > (11/6)\Delta$. This result implies that also the Glauber dynamics is rapidly mixing for $q > (11/6)\Delta$. Goldberg, Martin and Paterson~\cite{gmp-ssmfc-04} showed that any triangle free graph has strong spatial mixing provided $q>\alpha\Delta-\gamma$, where $\alpha$ is the solution to $\alpha^{\alpha}=e$ ($\alpha\approx1.76322$) and $\gamma=\frac{4\alpha^{3}-6\alpha^{2}-3\alpha+4}{2(\alpha^{2}-1)}\approx0.47031$. For triangle free graphs of low degree, this is still today the best general mixing bound that has been proved.

The technique Goldberg, Martin and Paterson used in~\cite{gmp-ssmfc-04} can be tailored and tweaked for particular graphs in order to give mixing bounds that are better than the general bound. This was demonstrated in~\cite{gmp-ssmfc-04} for the lattice $\mathbb{Z}^3$ with $q=10$ colours (the general result would give mixing for $q=11$ colours), and the triangular lattice with $q=10$ colours. Although the general result holds only when the graph is triangle free, the tailored proof for the triangular lattice does not require this. Both proofs were computer assisted, where the computational part consisted of looping though a huge number of boundary colourings and maximising certain values. This task would have been impossible to do by hand.

The general mixing bounds are tough barriers that seem difficult to break, though in many cases we expect mixing to occur with fewer colours. Therefore several proofs of mixing have been given for specific graphs or lattices. These proofs often involve computational assistance. We have mentioned two examples above as well as Salas and Sokal's 11-colour mixing bound on the triangular lattice~\cite{ss-apt-97}. Other examples of computer assisted proofs are those by Achlioptas, Molloy, Moore and van Bussel~\cite{ammb-sgcfc-04}, who showed mixing for $q=6$ colours on the grid, to which an alternative proof was given by Goldberg, Jalsenius, Martin and Paterson in~\cite{gjmp-imbafpm-06}. Salas and Sokal gave in~\cite{ss-apt-97} a computer assisted proof for $q=6$ colours on the kagome lattice, which was later improved to $q=5$ colours by Jalsenius in~\cite{Jal09:kagome}. It should also be mentioned that Jalsenius and Pedersen~\cite{JK08:grid-scan} have given a computer assisted proof of mixing with $q=7$ colours for the grid when the dynamics is updating vertices in deterministic order, as opposed to the Glauber dynamics which chooses a vertex at random in each step. The new colour is still chosen at random, though. This result is an improvement of the non-computer assisted proofs by Pedersen~\cite{p-dcssbd-07} and Dyer, Goldberg and Jerrum~\cite{dgj-dcss-06}.

The results on 9-colourings of the triangular lattice that we present in this paper are based on the 10-colour proof given by Goldberg, Martin and Paterson in~\cite{gmp-ssmfc-04}. Our 9-colour proof is of a whole different scale than the 10-colour proof and we must use computer assistance much more extensively and in more than one stage of the proof. The computations are rather demanding and prior to the final and rigourous results we had to use a heuristic to guide us in the right direction. We believe that the idea of such a heuristic could be useful to improve the mixing bounds for other lattices as well. However, our proof also demonstrates how demanding the computations can be, and unless new techniques are developed, there will probably be little progress in lowering the bounds for a vast number of lattices.


\section{Boundary pairs}
\label{sec:boundary-pairs}

 Similarly to Goldberg, Martin and Paterson in~\cite{gmp-ssmfc-04}, we define two structures referred to as \emph{vertex-boundary pairs} and \emph{edge-boundary pairs}. Before stating the formal definitions, we give an overview of the two concepts. A vertex-boundary pair consists of a region and two partial colourings of its vertex boundary. The colourings are identical except for on one boundary vertex, which is not allowed to have the colour~0 in either of the two colourings. That is, the vertex must have a ``real'' colour. An edge-boundary pair is similar to a vertex-boundary pair with the difference that the two boundary colourings are of the edge boundary instead of the vertex boundary. The formal definition of an edge-boundary pair might come across as slightly awkward as there are some additional conditions that must be met. The purpose of these conditions is to facilitate certain technicalities in the subsequent sections.

Formally, a \emph{vertex-boundary pair} $X$ consists of
\begin{itemize}
    \setlength{\itemsep}{0pt}
    \item a region $R_X$,
    \item a distinguished boundary vertex $w_X \in \partial R_X$         and
    \item a pair $(\calB_X,\calBprime_X)$ of partial 9-colourings
            of $\partial R_X$ such that
            \begin{itemize}
                \setlength{\itemsep}{0pt}
                \item $\calB_X(v)=\calBprime_X(v)$ for all
                        $v\in \partial R_X \setminus \{w_X\}$,
                        $\calB_X(w_X)\neq \calBprime_X(w_X)$
                        and
                \item $\calB_X(w_X)> 0$ and
                        $\calBprime_X(w_X)> 0$.
            \end{itemize}
\end{itemize}
An \emph{edge-boundary pair} $X$ consists of
\begin{itemize}
    \setlength{\itemsep}{0pt}
    \item a region $R_X$,
    \item a distinguished boundary edge $e_X = \{w_X,v_X\}\in \calE
            R_X$, where $w_X\in \partial R_X$ and $v_X\in R_X$
            are two distinguished vertices such that at most \emph{five}
            neighbours of $w_X$ are in $R_X$, and
    \item a pair $(B_X,\Bprime_X)$ of partial 9-colourings of $\calE R_X$
            such that
        \begin{itemize}
            \setlength{\itemsep}{0pt}
            \item $B_X(e)=\Bprime_X(e)$ for all
                    $e\in \calE R_X \setminus \{e_X\}$, $B_X(e_X)\neq
                    \Bprime_X(e_X)$,
            \item $B_X(e_X)> 0$, $\Bprime_X(e_X)> 0$, and
            \item for any two clockwise adjacent edges $e_1,e_2\in
                    \calE R_X$ that share a vertex $w\in \partial R$,
                    $B_X(e_1)=B_X(e_2)$ or $\Bprime_X(e_1)=\Bprime_X(e_2)$.
        \end{itemize}
\end{itemize}
Note that the very last condition means that two clockwise adjacent edges have the same colour in both $B_X$ and $\Bprime_X$ unless one of the edges is $e_X$.

We let $E_X=\{\{v_X,u\} \mid \{v_X,u\}\in
E_\trilat \text{ and } u\in R_X\}$ denote the set of edges between $v_X$ and a vertex in
$R_X$.

For a vertex-boundary pair $X$ (or edge-boundary pair $X$),
a coupling $\Psi$ of $\pi_{\calB_X}^9$ and $\pi_{\calBprime_X}^9$
(or $\pi_{B_X}^9$ and $\pi_{\Bprime_X}^9$) and $v\in R_X$, we define the indicator random variable
\begin{equation*}
1_{\Psi,v}=
\begin{cases}
     1, &\text{$(\calC,\calC')$ is a pair of colourings drawn
        from $\Psi$ and $\calC(v)\neq \calC'(v)$;}\\
     0, &\text{$\calC(v)= \calC'(v)$\,.}
\end{cases}
\end{equation*}

For an edge-boundary pair $X$, we define $\Psi_X$ to be some coupling $\Psi$ of $\pi_{B_X}^9$ and $\pi_{\Bprime_X}^9$ minimising $\expect{1_{\Psi,v_X}}$. We define $\mu(X)=\expect{1_{\Psi_X,v_X}}$. For every pair $(c,c')\in [9]\times [9]$ of colours, we let $p_X(c,c')$ be the probability that, when a pair $(\calC,\calC')$ of colourings is drawn from $\Psi_X$, $\calC(v_X)=c$ and $\calC'(v_X)=c'$.
Note that
\begin{equation*}
    \mu(X)= \expect{1_{\Psi_X,v_X}}= \mathop{\sum_{c,c'\in [9]}}_{c\neq c'} p_X(c,c')\,.
\end{equation*}
%


\section{Recursive coupling}
\label{sec:recursive}

In order to show strong spatial mixing for $q=9$ colours, we show that for all vertex-boundary pairs $X$ and subregions $R'\subseteq R_X$, there exists a coupling $\Psi$ of $\pi_{\calB_X}^9$ and $\pi_{\calBprime_X}^9$ such that
\begin{equation*}
    \sum_{v\in R'} \expect{1_{\Psi,v}}
\end{equation*}
decreases exponentially in the distance between $w_X$ and $R'$. As we will see, this implies strong spatial mixing for $q=9$ colours. In order to show exponential decay, it will be convenient to work with edge-boundary pairs.

We closely follow the approach taken by Goldberg, Martin and Paterson in~\cite{gmp-ssmfc-04} and define a tree
$T_X$ associated with each edge-boundary pair $X$. The tree $T_X$ is constructed as follows (Figure~\ref{fig:tree} illustrates an example of a tree $T_X$).
\begin{figure}[p]
    \centering
    \showgraphics{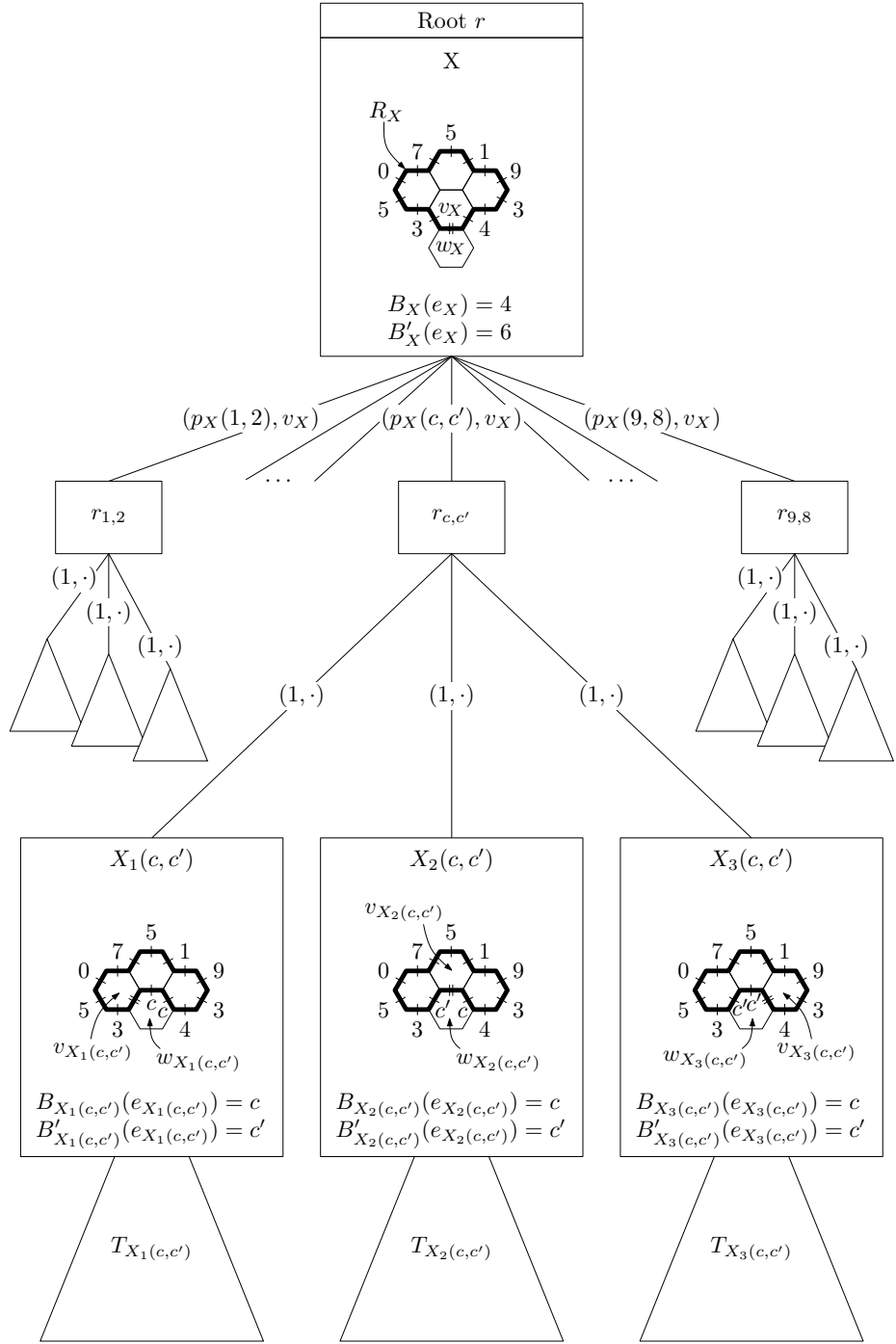}
    \caption{\label{fig:tree}An example of the tree $T_X$.}
\end{figure}
Start with a node $r$ which will be the root of $T_X$. For every pair $(c,c')\in [9]\times [9]$ of colours such that $c\neq c'$, add an edge labelled $(p_X(c,c'),v_X)$ from $r$ to a new node $r_{c,c'}$. If $E_X$ is empty, $r_{c,c'}$ is a leaf. Otherwise, let $e_1,\dots,e_k$ be the edges in $E_X$ such that $e_X\prec e_1\prec \cdots \prec e_k$. For each $i\in \{1,\dots,k\}$, let $X_i(c,c')$ be the edge-boundary pair consisting of
\begin{itemize}
    \setlength{\itemsep}{0pt}
    \item the region $R_{X_i(c,c')}=R_X\setminus \{v_X\}$,
    \item the distinguished boundary edge $e_{X_i(c,c')}= \{w_{X_i(c,c')},v_{X_i(c,c')}\}= e_i$, where $w_{X_i(c,c')}= v_X$, and
    \item the pair $(B_{X_i(c,c')},\Bprime_{X_i(c,c')})$ of partial 9-colourings of $\calE R_{X_i(c,c')}$
            such that
        \begin{itemize}
            \setlength{\itemsep}{0pt}
            \item $B_{X_i(c,c')}(e)= B_X(e)$ for $e\in \calE R_{X_i(c,c')}\cap \calE R_X$,
            \item $B_{X_i(c,c')}(e)= c'$ for $e\in \{e_1,\dots,e_{i-1}\}$,
            \item $B_{X_i(c,c')}(e)= c$ for $e\in \{e_i,\dots,e_k\}$, and
            \item $\Bprime_{X_i(c,c')}$ is identical to $B_{X_i(c,c')}$ on all edges but $e_i$ for which $\Bprime_{X_i(c,c')}(e_i)= c'$.
        \end{itemize}
\end{itemize}

Note that the properties of $X_i(c,c')$ meet all the requirements for being an edge-boundary pair; the vertex $w_{X_i(c,c')}$ has at most five neighbours in $R_{X_i(c,c')}$ (since $w_X\notin R_{X_i(c,c')}$ is a neighbour of $w_{X_i(c,c')}$), and any two clockwise adjacent edges in $\calE R_{X_i(c,c')}$ that share a vertex in $\partial R_{X_i(c,c')}$ have the same colour in $B_{X_i(c,c')}$ or $\Bprime_{X_i(c,c')}$ (or both).

Recursively construct $T_{X_i(c,c')}$, the tree corresponding to edge-boundary pair $X_i(c,c')$. Add an edge with label $(1,\cdot)$ from $r_{c,c'}$ to the root of $T_{X_i(c,c')}$. That completes the construction of $T_X$.

We say that an edge $e$ of $T_{X}$ is \emph{degenerate} if the second component of its label is `$\cdot$'. For edges $e$ and $e'$ of $T_{X}$, we write $e\rightarrow e'$ to denote the fact that $e$ is an ancestor of $e'$. That is, either $e=e'$, or $e$ is a proper ancestor of $e'$. Define the \emph{level} of edge $e$ to be the number of non-degenerate edges on the path from the root down to, and including, $e$. Suppose that $e$ is an edge of $T_{X}$ with label $(p,v)$. We say that the \emph{weight} $w(e)$ of edge $e$ is $p$. Also the \emph{name} $n(e)$ of edge $e$ is $v$. The \emph{likelihood} $\ell(e)$ of $e$ is $\prod_{e':e'\rightarrow e}w(e)$. The \emph{cost} $\gamma(v,T_{X})$ of a vertex $v\in R_{X}$ in $T_{X}$ is $\sum_{e:n(e)=v}\ell(e)$.

\begin{lemma}[Lemma~12 of Goldberg et al.\@ \cite{gmp-ssmfc-04}]
    \label{lem:cost}
    For every edge-boundary pair $X$ there exists a coupling $\Psi$ of $\pi_{B_{X}}^9$ and $\pi_{\Bprime_{X}}^9$ such that, for all $v\in R_{X}$, $\expect{1_{\Psi,v}}\leq \gamma(v,T_{X})$.
\end{lemma}

In the proof, given by Goldberg, Martin and Paterson in~\cite{gmp-ssmfc-04}, the coupling $\Psi$ is constructed recursively in the same manner as the tree $T_X$. The discrepancy at a given boundary vertex is broken to discrepancies at single boundary edges, so at every stage of the recursion, only pairs of colourings with a discrepancy at a single edge have to be considered (i.e., edge-boundary pairs).

For an edge-boundary pair $X$ and $d\geq 1$, we let $E_d(X)$ denote the set of level-$d$ edges in $T_X$. We define $\Gamma_d(X)= \sum_{e\in E_d(X)} \ell(e)$.

\begin{lemma}
    \label{lem:gamma-bounded}
    For every edge-boundary pair $X$ and $R\subseteq R_X$ there exists a coupling $\Psi$ of $\pi_{B_{X}}^9$ and $\pi_{\Bprime_{X}}^9$ such that
    \begin{equation*}
        \sum_{v\in R} \expect{1_{\Psi,v}}\leq
            \sum_{d\geq d_{R_X}\!(w_X,R)} \Gamma_d(X)\,.
    \end{equation*}
\end{lemma}

\begin{proof}
By Lemma~\ref{lem:cost} there is a coupling $\Psi$ such that
\begin{align*}
    \sum_{v\in R} \expect{1_{\Psi,v}} \,&\leq\,
        \sum_{v\in R} \gamma(v,T_{X}) \,=\,
             \sum_{v\in R} \, \sum_{e:n(e)=v} \ell(e) \,\leq
                \sum_{d\geq d_{R_X}\!(w_X,R)}\,
                \sum_{e\in E_d(X)} \ell(e)\\
    &= \sum_{d\geq d_{R_X}\!(w_X,R)} \Gamma_d(X)\,.
    \qedhere
\end{align*}
\end{proof}

The following recursive definition of $\Gamma_d(X)$ is equivalent to the definition above and will be useful in subsequent sections.
\begin{equation*}
\Gamma_d(X)=
\begin{cases}
    \displaystyle \mathop{\sum_{c,c'\in [9]}}_{c\neq c'} p_X(c,c'),
                &d=1\,; \vspace{2mm} \\
    \displaystyle \mathop{\sum_{c,c'\in [9]}}_{c\neq c'} \left(p_X(c,c')
                \sum_{i=1}^{|E_X|}
                    \Gamma_{d-1}(X_i(c,c'))\right),
                &d>1\,.\\
\end{cases}
\end{equation*}
Note that $\Gamma_1(X)= \mu(X)$.
For a set $U$ of edge-boundary pairs, we define $\Gamma_d(U)= \max_{X\in U} \Gamma_d(X)$. We define $\Gamma_d(\emptyset)=0$ for all $d$.


\section{Exponential decay}
\label{sec:exponential-decay}

Our first step towards a proof of exponential decay and strong spatial mixing is to show that for any edge-boundary pair $X$, $\Gamma_d(X)$ decreases exponentially with $d$. A key ingredient in the proof is the quantity $\mu(X)$, for which we want to derive sufficiently good upper bounds. This is where we start.

We use the lemma below by Goldberg, Martin and Paterson~\cite{gmp-ssmfc-04}\footnote{\,In~\cite{gmp-ssmfc-04}, Goldberg, Martin and Paterson define $\nu(X)=\expect{1_{\Psi_X,v_X}}$, which is the definition of $\mu(X)$ in this article. They give an alterative definition of $\mu(X)$, however, as pointed out in the proof of Lemma~13 in~\cite{gmp-ssmfc-04}, their $\mu(X)= \nu(X)$ indeed.}. The idea is to shrink the region $R_X$ so that $v_X$ is kept within the smaller region, and use this smaller region to construct a new edge-boundary pair $X'$ whose boundary colourings are identical to the boundary colourings of $X$ on overlapping boundary edges. The colours of the boundary edges introduced by shrinking $R_X$ are chosen to maximise $\mu(X')$. Then $\mu(X)\leq \mu(X')$.

\begin{lemma}[Lemma~13 of Goldberg et al.\@ \cite{gmp-ssmfc-04}]
    \label{lem:convexity}
    Suppose that $X$ is an edge-boundary pair. Let $R$ be any subset of $R_X$ which includes $v_X$. Let $\Lambda$ be the set of edge-boundary pairs $X'$ such that $R_{X'}= R$, $e_{X'}= e_X$, and, for $e\in \calE R_X\cap \calE R_{X'}$, $B_{X'}(e)= B_X(e)$ and $\Bprime_{X'}(e)= \Bprime_X(e)$. Then $\mu(X)\leq \max_{X'\in \Lambda}\mu(X')$.
\end{lemma}

In this article, the 39 regions $M_i$ illustrated in Figure~\ref{fig:M-regions} are of particular importance to us. In Section~\ref{sec:experimental} we discuss why we chose these regions.%
\begin{figure}[p]
    \centering
    \renewcommand{\arraystretch}{1.00}
    \newcommand{\diag}[1]{\showgraphics{m-regions/m#1}}
    \newcommand{\dlab}[1]{$M_{#1}$\vspace{0.10cm}}
    \begin{tabular*}{1.00\textwidth}{@{\extracolsep{\fill}} c c c c}
        \diag{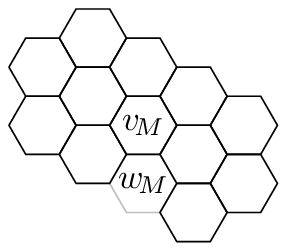} & \diag{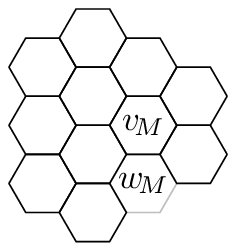} & \diag{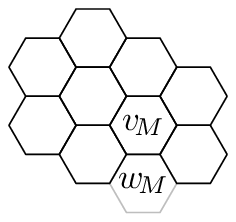} & \diag{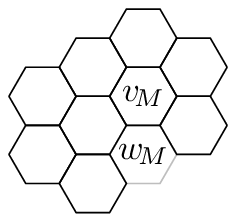} \\
        \dlab{1}  & \dlab{2}  & \dlab{3}  & \dlab{4} \\
        \diag{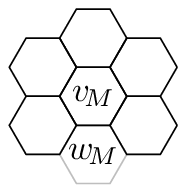} & \diag{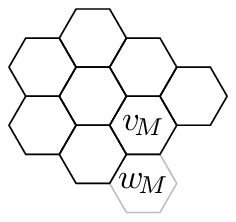} & \diag{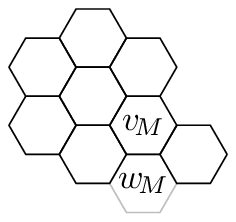} & \diag{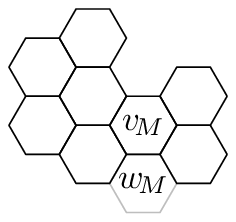} \\
        \dlab{5}  & \dlab{6}  & \dlab{7}  & \dlab{8} \\
        \diag{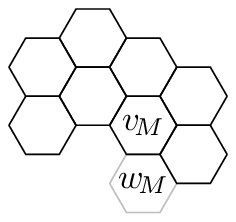} & \diag{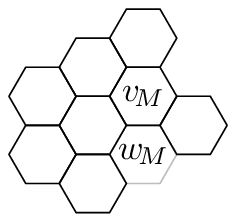} & \diag{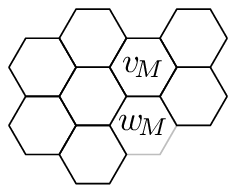} & \diag{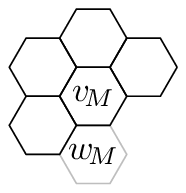} \\
        \dlab{9}  & \dlab{10} & \dlab{11} & \dlab{12} \\
        \diag{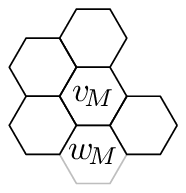} & \diag{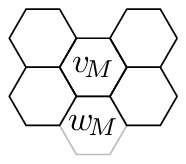} & \diag{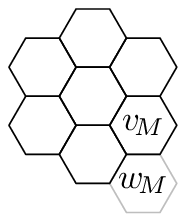} & \diag{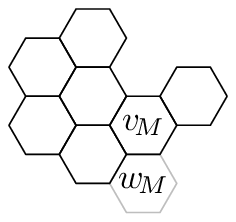} \\
        \dlab{13} & \dlab{14} & \dlab{15} & \dlab{16} \\
        \diag{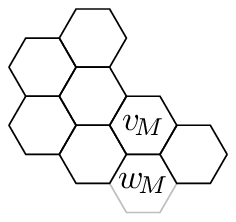} & \diag{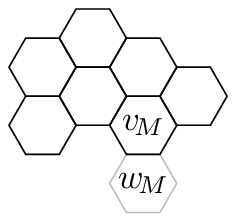} & \diag{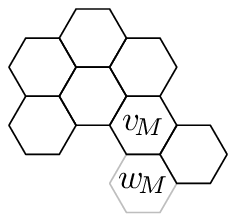} & \diag{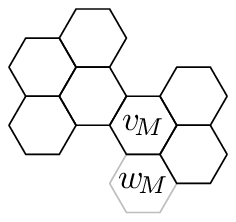} \\
        \dlab{17} & \dlab{18} & \dlab{19} & \dlab{20} \\
    \end{tabular*}

    \bigskip

    Part 1 of 2
    \caption{\label{fig:M-regions}
        The regions $M_1,\dots,M_{39}$.}
\end{figure}
\begin{figure}[p]
    \addtocounter{figure}{-1}
    \centering
    \renewcommand{\arraystretch}{1.00}
    \newcommand{\diag}[1]{\showgraphics{m-regions/m#1}}
    \newcommand{\dlab}[1]{$M_{#1}$\vspace{0.10cm}}
    \begin{tabular*}{1.00\textwidth}{@{\extracolsep{\fill}} c c c c}
        \diag{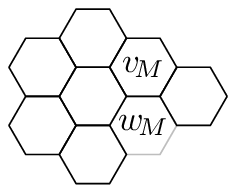} & \diag{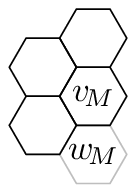} & \diag{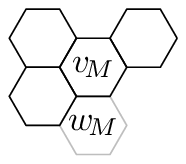} & \diag{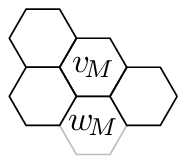} \\
        \dlab{21} & \dlab{22} & \dlab{23} & \dlab{24} \\
        \diag{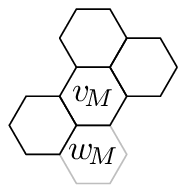} & \diag{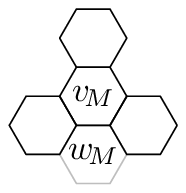} & \diag{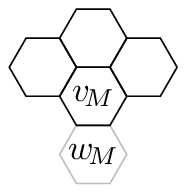} & \diag{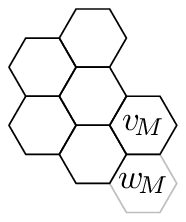} \\
        \dlab{25} & \dlab{26} & \dlab{27} & \dlab{28} \\
        \diag{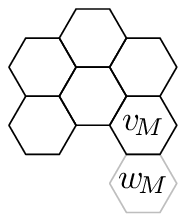} & \diag{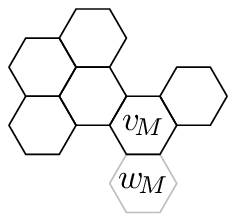} & \diag{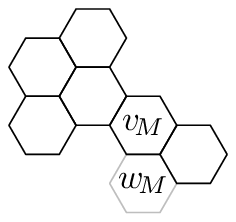} & \diag{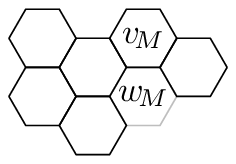} \\
        \dlab{29} & \dlab{30} & \dlab{31} & \dlab{32} \\
        \diag{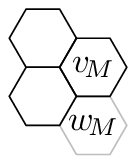} & \diag{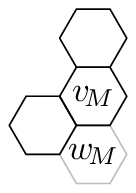} & \diag{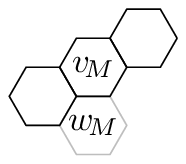} & \diag{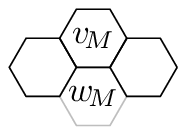} \\
        \dlab{33} & \dlab{34} & \dlab{35} & \dlab{36} \\
        \diag{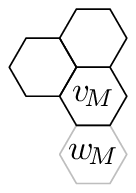} & \diag{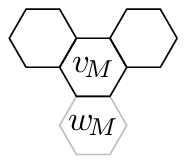} & \diag{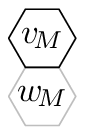} & ~ \\
        \dlab{37} & \dlab{38} & \dlab{39} & ~ \\
    \end{tabular*}

    \bigskip

    Part 2 of 2
    \caption{The regions $M_{1},\dots,M_{39}$.}
\end{figure}
In order to upper-bound $\mu(X)$ for an arbitrary edge-boundary pair $X$, we will shrink the region $R_X$ down to one of the 39 regions and apply Lemma~\ref{lem:convexity}. In order to successfully shrink $R_X$ down to match an $M$-region, we might have to consider an appropriate rotation or reflection of the region.

For $i\in [39]$, we define the constants $\mu_i$ in Table~\ref{tab:mu-values} and prove the following lemma with the help of a computer. Details of the proof are given in Section~\ref{sec:computations}.

\begin{lemma}
    \label{lem:mu-values}
    For $i\in [39]$, $\mu(X)\leq \mu_i$ for every edge-boundary pair $X$ such that $R_X$ is the region $M_i$ in Figure~\ref{fig:M-regions} and $v_X$ and $w_X$ are the vertices labelled $v_M$ and $w_M$, respectively.
\end{lemma}

\begin{table}[t]
    \centering\renewcommand{\arraystretch}{1.15}
    \begin{tabular}{l l l}
    $\mu_{1}=68809973/310505657$ & $\mu_{14}=25/91$ &
                                            $\mu_{27}=21/73$ \\
    $\mu_{2}=11623551/51797443$ & $\mu_{15}=9334/40215$ &
                                            $\mu_{28}=2833/11551$ \\
    $\mu_{3}=456459/2005687$ & $\mu_{16}=11332/46633$ &
                                            $\mu_{29}=2833/11551$ \\
    $\mu_{4}=408609/1601573$ & $\mu_{17}=11332/46633$ &
                                            $\mu_{30}=620/2321$ \\
    $\mu_{5}=33/127$ & $\mu_{18}=7067/29188$ &
                                            $\mu_{31}=620/2321$ \\
    $\mu_{6}=18199/78779$ & $\mu_{19}=11332/46633$ &
                                            $\mu_{32}=688/2389$ \\
    $\mu_{7}=75312/325193$ & $\mu_{20}=775/2941$ &
                                            $\mu_{33}=5/17$ \\
    $\mu_{8}=14165/58613$ & $\mu_{21}=4248/16015$ &
                                            $\mu_{34}=4/13$ \\
    $\mu_{9}=70661/293514$ & $\mu_{22}=21/73$ &
                                            $\mu_{35}=4/13$ \\
    $\mu_{10}=31648/123341$ & $\mu_{23}=5/17$ &
                                            $\mu_{36}=4/13$ \\
    $\mu_{11}=2655/10063$ & $\mu_{24}=5/17$ &
                                            $\mu_{37}=5/17$ \\
    $\mu_{12}=521/1853$ & $\mu_{25}=5/17$ &
                                            $\mu_{38}=4/13$ \\
    $\mu_{13}=208/757$ & $\mu_{26}=32/113$ &
                                            $\mu_{39}=1/3$
\end{tabular}
\caption{
    $\mu_i$ is $\mu(X)$ maximised over all edge-boundary pairs $X$ whose region $R_X$ is the region $M_i$ in Figure~\ref{fig:M-regions}.}
\label{tab:mu-values}
\end{table}


\subsection{Regions and sets of edge-boundary pairs}


Let $F$ be the region in Figure~\ref{fig:f}, where a vertex $v_F\in F$ and $w_F\in \partial F$ are labelled. Since $F$ contains 12 vertices, we define $F_1,\dots,F_{2048}$ to be the $2^{11}=2048$ distinct subregions of $F$ that all contain the vertex labelled $v_F$. For $i\in [2048]$, we define $U_i$ to be the set of edge-boundary pairs $X$ such that the intersection of $R_X$ and $F$ is $F_i$, where $v_X$ and $w_X$ coincide with the vertices labelled $v_F$ and $w_F$, respectively. Thus, for any edge-boundary pair $X$, there is a unique $i\in[2048]$ such that $X\in U_i$.
\begin{figure}[t]
    \begin{minipage}[t]{0.47\linewidth}
        \centering
        \showgraphics{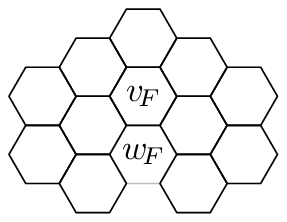}
        \caption{\label{fig:f}The region $F$.}
    \end{minipage}
    \hfill
    \begin{minipage}[t]{0.47\linewidth}
        \centering
        \showgraphics{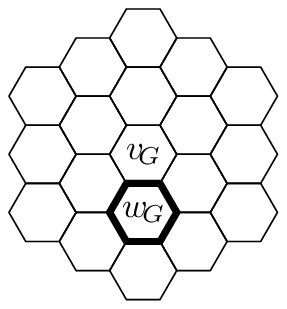}
        \caption{\label{fig:g}The region $G$. The vertex $w_G$ does not belong to $G$.}
    \end{minipage}
\end{figure}

Let $G$ be the region in Figure~\ref{fig:g}, where a vertex $v_G\in G$ and $w_G\in \partial G$ are labelled. The vertex $w_G$ is a ``hole'' in $G$. We define $\calG$ to be the set of subregions $G'$ of $G$ such that $G'$ contains $v_G$ and at least one neighbour of $w_G$ is not in $G'$.
Recall that in the definition of an edge-boundary pair $X$, $w_X$ has at most five neighbours in $R_X$.

We define a function $\Phi:\calG\rightarrow [39]\times \{0,\dots,2048\}^6$.
Suppose $G'\in \calG$ is a region. Then $\Phi(G')= (m,b_0,\dots,b_5)$ where $m$, $b_0,\dots,b_5$ are uniquely specified as follows.
Let $\calM\subseteq \{M_1,\dots,M_{39}\}$ be the set of regions $M_i$ such that $M_i$ (or the reflection of $M_i$) is a subregion of $G'$, where $v_G$ and $w_G$ coincide with $v_M$ and $w_M$, respectively. $M_m$ is a region in $\calM$ such that $\mu_m\leq \mu_i$ for all $M_i\in \calM$.
When well defined,
$F_{b_0},\dots,F_{b_5}$ are the intersections of $G'$ and $F$ taken according to Figure~\ref{fig:intersect-a},\dots,\ref{fig:intersect-f}, respectively.
If $F_{b_i}$ is not well defined then we set $b_i=0$.
For example, if the vertex above $v_G$ in region $G$ (Figure~\ref{fig:g}) is not in $G'$ then none of the regions $F_1,\dots,F_{2048}$ is the intersection of $G'$ and $F$ in Figure~\ref{fig:intersect-d}, hence $b_3=0$.
\begin{figure}[t]
  \centering
  ~
  \hfill
  \subfloat[][]{\label{fig:intersect-a}\showgraphics{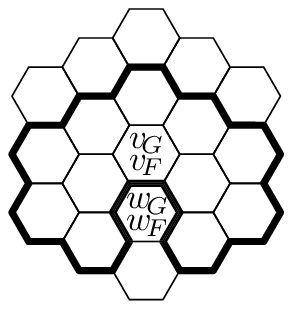}}
  \hfill
  \subfloat[][]{\label{fig:intersect-b}\showgraphics{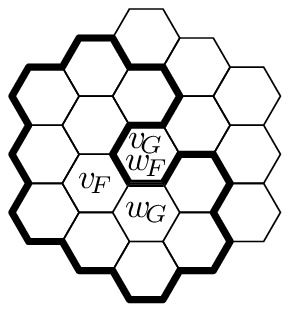}}
  \hfill
  \subfloat[][]{\label{fig:intersect-c}\showgraphics{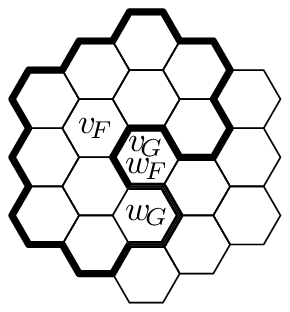}}
  \hfill
  ~

  ~
  \hfill
  \subfloat[][]{\label{fig:intersect-d}\showgraphics{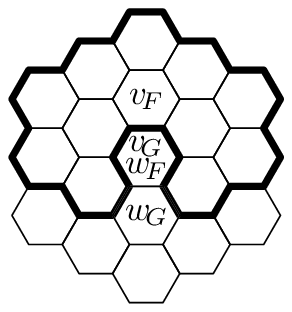}}
  \hfill
  \subfloat[][]{\label{fig:intersect-e}\showgraphics{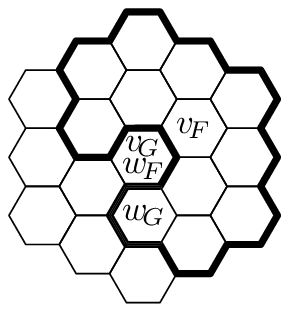}}
  \hfill
  \subfloat[][]{\label{fig:intersect-f}\showgraphics{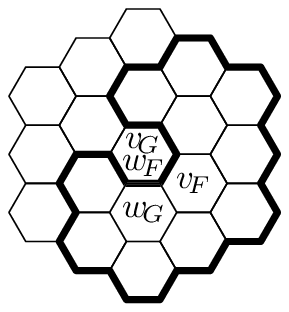}}
  \hfill
  ~
  \caption{Intersections of region $G$ and $F$.}
  \label{fig:intersections}
\end{figure}


\subsection{Upper bounds}

We use a computer to prove the next lemma.

\begin{lemma}
    \label{lem:alphas}
    There exist constants $\alpha_1,\dots,\alpha_{2048}\in [2,6]$ such that, for every region $G'\in \calG$,
    \begin{equation}
        \label{eq:inequality}
        \mu_m(\alpha_{b_1}+\cdots+\alpha_{b_5})\leq
            \alpha_{b_0}(1-\epsilon)\,,
    \end{equation}
    where $(m,b_0,\dots,b_5)=\Phi(G')$, $\alpha_0=0$, $\epsilon=1/1000$ and the values $\mu_m$ are from Table~\ref{tab:mu-values}.
\end{lemma}

The constants $\alpha_i$ in the lemma above have been obtained in the following way. We have written a computer program which goes through the regions $G'$ of $\calG$, calculates the values $m,b_0,\dots,b_5$ and adds Equation~(\ref{eq:inequality}) above to a linear program. Once all regions in $\calG$ have been processed, the linear program contains thousands of inequalities. We then use a linear program solver to successfully find a satisfying solution such that each value $\alpha_i$ is in the interval $[2,6]$. See Section~\ref{sec:computations} for more information on the computational part.

\begin{lemma}
    \label{lem:gamma-decay}
    There exist constants $\alpha_1,\dots,\alpha_{2048}\in [2,6]$ such that,
    for $i\in [2048]$ and $d\geq 1$, $\Gamma_d(U_i)\leq \alpha_i(1-\epsilon)^d$, where $\epsilon= 1/1000$.
\end{lemma}

\begin{proof}
    Let $\alpha_1,\dots,\alpha_{2048}\in [2,6]$ be constants satisfying the inequalities in Lemma~\ref{lem:alphas}.
    For $i\in [2048]$, suppose $X\in U_i$. We use induction on $d$ to show that $\Gamma_d(X)\leq \alpha_i(1-\epsilon)^d$.

    For the base case $d=1$, $\Gamma_1(X)= \mu(X)\leq 1< 2(1-\epsilon)\leq \alpha_i(1-\epsilon)$.

    For the inductive step $d>1$, let $G'$ be the intersection of $R_X$ and $G$ such that $v_X$ and $w_X$ coincide with $v_G$ and $w_G$, respectively.
    Hence $G'\in \calG$.

    Let $(m,b_0,\dots,b_5)= \Phi(G')$ and note that $b_0=i$. Since $M_m$ is a subregion of $R_X$, it follows from Lemmas~\ref{lem:convexity} and~\ref{lem:mu-values} that $\mu(X)\leq \mu_m$.

    Recall that $E_X$ is the set of edges between $v_X$ and a vertex in $R_X$. For $r\in \{1,\dots,|E_X|\}$ and $(c,c')\in [9]\times [9]$ such that $c\neq c'$, let $X_r(c,c')$ be the edge-boundary pairs defined in the construction of the tree $T_X$ in Section~\ref{sec:recursive}.
    From the definition of $\Phi$, the intersection of $R_{X_r(c,c')}$ and $F$, taken such that $v_{X_r(c,c')}$ and $w_{X_r(c,c')}$ coincide with $v_F$ and $w_F$, respectively, is $F_{b_k}$, where the value of $k$ depends on the edge $e_r\in E_X$. Thus, with $U_0=\emptyset$,
    \begin{equation}
        \label{eq:U-sets}
        \sum_{r=1}^{|E_X|} \Gamma_{d-1}(X_r(c,c'))\leq
            \sum_{k=1}^{5} \Gamma_{d-1}(U_{b_k})\,.
    \end{equation}

    We have
    \begin{align*}
    \Gamma_d(X) &=      \mathop{\sum_{c,c'\in [9]}}_{c\neq c'} p_X(c,c')
                        \sum_{r=1}^{|E_X|}
                        \Gamma_{d-1}(X_r(c,c'))
                            &~ \text{(Definition of $\Gamma_d(X)$)}\\
                &\leq   \mathop{\sum_{c,c'\in [9]}}_{c\neq c'} p_X(c,c')
                        \sum_{k=1}^{5} \Gamma_{d-1}(U_{b_k})
                            &~ \text{(Equation~(\ref{eq:U-sets}))}\\
                &=      \mu(X)
                        \sum_{k=1}^{5} \Gamma_{d-1}(U_{b_k})
                            &~ \text{(Definition of $\mu(X)$)}\\
                &\leq   \mu(X)
                        \sum_{k=1}^{5} \alpha_{b_k}(1-\epsilon)^{d-1}
                            &~ \text{(Induction hypothesis)}\\
                &\leq   \mu_m (\alpha_{b_1}+\dots + \alpha_{b_5})
                        (1-\epsilon)^{d-1}
                            &~ \text{($\mu(X)\leq \mu_m$)}\\
               &\leq   \alpha_{i}(1-\epsilon)^d\,.
                            &~ \text{(Lemma~\ref{lem:alphas} and $b_0=i$)}
    \end{align*}
\end{proof}

The next corollary follows immediately from Lemma~\ref{lem:gamma-decay}.

\begin{corollary}
    \label{cor:gamma-decay}
    For every edge-boundary pair $X$ and $d\geq 1$,
    $\Gamma_d(X)\leq 5(1-\epsilon)^d$, where $\epsilon= 1/1000$.
\end{corollary}

\begin{lemma}
    \label{lem:edge-decay}
    For every edge-boundary pair $X$ and subregion $R\subseteq R_X$ there exists a coupling $\Psi$ of $\pi^9_{B_X}$ and $\pi^9_{B'_X}$ such that
    \begin{equation*}
        \sum_{v\in R} \expect{1_{\Psi_i,v}} \leq
            \frac{5}{\epsilon} (1-\epsilon)^{d_{R_X}\!(w_X,R)}\,,
    \end{equation*}
    where $\epsilon = 1/1000$.
\end{lemma}
\begin{proof}
    Follows from Lemma~\ref{lem:gamma-bounded} and Corollary~\ref{cor:gamma-decay}.
\end{proof}


\subsection{Vertex-boundary pairs and strong spatial mixing}

Similarly to Lemma~\ref{lem:edge-decay}, we prove the following lemma for vertex-boundary pairs.

\begin{lemma}
    \label{lem:vertex-decay}
    For every vertex-boundary pair $X$ and subregion $R\subseteq R_X$ there exists a coupling $\Psi$ of $\pi_{\calB_{X}}^9$ and $\pi_{\calBprime_{X}}^9$ such that
    \begin{equation*}
        \sum_{v\in R} \expect{1_{\Psi,v}} \leq
            \frac{50}{\epsilon (1-\epsilon)}
            (1-\epsilon)^{d_{R_X}\!(w_X,R)}\,,
    \end{equation*}
    where $\epsilon = 1/1000$.
\end{lemma}

\begin{proof}
    Let $X$ be any vertex-boundary pair and let $R\subseteq R_X$. First suppose that $w_X$ has a neighbour $y\notin R_X$. Let $E=\{e_1,\dots,e_k\}\subseteq \calE R_X$ be the boundary edges incident to $w_X$ such that $\{w_X,y\}\prec e_1\prec \cdots \prec e_k$.
    For $i\in [k]$, let $X_i$ be the edge-boundary pair consisting of
    \begin{itemize}
        \setlength{\itemsep}{0pt}
        \item the region $R_{X_i}=R_X$,
        \item the distinguished boundary edge
                $e_{X_i}= e_i$,
        \item the pair $(B_{X_i},\Bprime_{X_i})$ of partial
                9-colourings of $\calE R_{X_i}$
                such that
            \begin{itemize}
                \setlength{\itemsep}{0pt}
                \item $B_{X_i}(\{w,v\})= \Bprime_{X_i}(\{w,v\})=
                        \calB_X(w)$ for $\{w,v\}\in \calE
                        R_X\setminus E$, where $w\in \partial R_X$,
                \item $B_{X_i}(e_j)=
                        \Bprime_{X_i}(e_j)= \calBprime_X(w_X)$
                        for $j\in \{1,\dots,i-1\}$,
                \item $B_{X_i}(e_i)= \calB_X(w_X)$ and
                        $\Bprime_{X_i}(e_i)= \calBprime_X(w_X)$, and
                \item $B_{X_i}(e_j)=
                        \Bprime_{X_i}(e_j)= \calB_X(w_X)$
                        for $j\in \{i+1,\dots,k\}$.
            \end{itemize}
    \end{itemize}
    Note that $\pi_{\calB_X}^9= \pi_{B_{X_1}}^9$ and $\pi_{\calBprime_X}^9= \pi_{\Bprime_{X_k}}^9$. Figure~\ref{fig:vertex-to-edges} illustrates an example of how the edge-boundary pairs $X_i$ are constructed.
    \begin{figure}[t]
        \centering
        \showgraphics{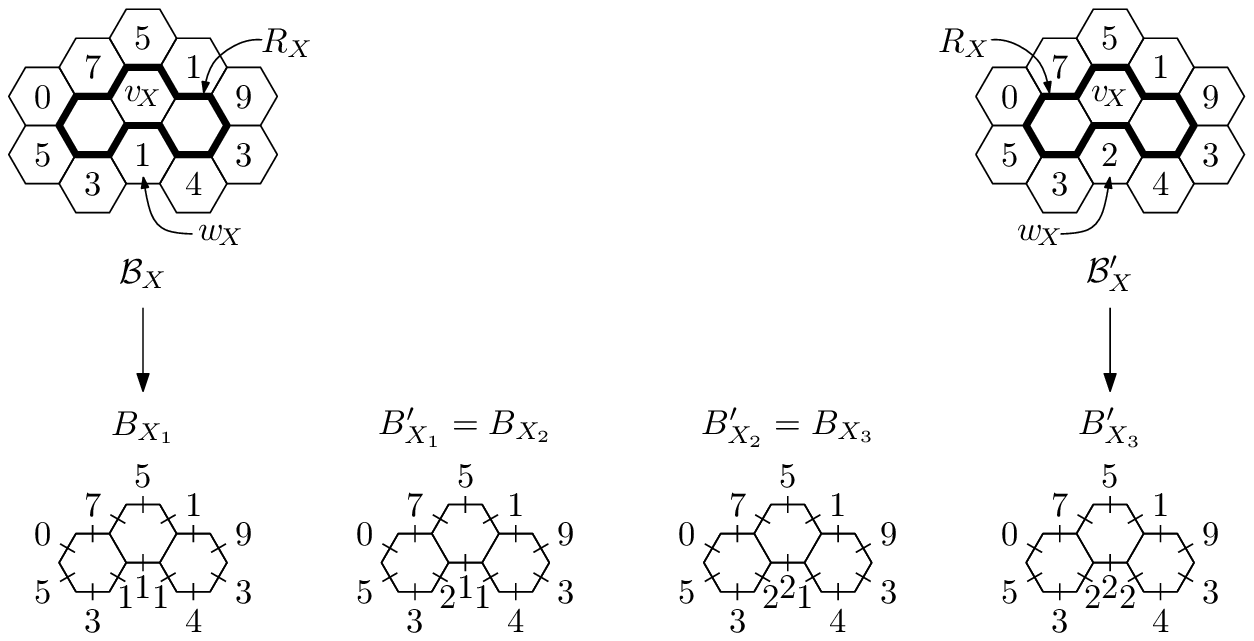}
        \caption{\label{fig:vertex-to-edges}A vertex-boundary pair $X$ broken into three edge-boundary pairs $X_1$, $X_2$ and $X_3$. The numbers represent the colours of the boundary vertices and boundary edges, respectively.}
    \end{figure}

    We use Lemma~\ref{lem:edge-decay} and for $i\in [k]$ we let $\Psi_i$ be a coupling of
    $\pi_{B_{X_i}}^9$ and $\pi_{\Bprime_{X_i}}^9$ such that
    \begin{equation}
        \label{eq:exp-decay}
        \sum_{v\in R} \expect{1_{\Psi_i,v}} \leq
            \frac{5}{\epsilon} (1-\epsilon)^{d_{R_X}\!(w_X,R)}\,.
    \end{equation}

    We define a coupling $\Psi$ of $\pi_{\calB_X}^9$ and $\pi_{\calBprime_X}^9$ by composing the couplings $\Psi_1,\dots,\Psi_k$ as follows.
    Let $(\calC_0,\calC_1),(\calC_1,\calC_2),\dots,\nolinebreak[3](\calC_{k-1},\calC_k)$ be pairs of colourings drawn from $\Psi_1,\dots,\Psi_k$, respectively. Then $(\calC_0,\calC_k)$ is the pair of colourings drawn from $\Psi$.
    If, for $v\in R$, $\calC_0(v)\neq \calC_k(v)$, then $\calC_{j-1}(v)\neq \calC_j(v)$ for some $j\in [k]$.
    Hence
    \begin{equation}
        \label{eq:prob-sum}
        \expect{1_{\Psi,v}} \leq
            \sum_{i=1}^k \expect{1_{\Psi_i,v}}\,.
    \end{equation}

    From Equations~(\ref{eq:exp-decay}) and~(\ref{eq:prob-sum}), and $k$ being at most~5, we have
    \begin{equation*}
        \sum_{v\in R} \expect{1_{\Psi,v}} \leq
        \sum_{v\in R} \sum_{i=1}^k \expect{1_{\Psi_i,v}} =
        \sum_{i=1}^k \sum_{v\in R} \expect{1_{\Psi_i,v}} \leq
        \frac{25}{\epsilon} (1-\epsilon)^{d_{R_X}\!(w_X,R)}\,.
    \end{equation*}

    Lastly, suppose that all neighbours of $w_X$ are in $R_X$.
    A technical detail arises here because we can no longer
    break the discrepancy at $w_X$ into edge-boundary pairs $X_i$ as above. Instead we will first randomly choose a colour of a neighbour $u$ of $w_X$ and then define edge-boundary pairs for the region $R_X\setminus \{u\}$.

    Let $u$ be any neighbour of $w_X$ and let $R_u= R_X\setminus \{u\}$.
    We define a coupling $\Psi$ of $\pi_{\calB_X}^9$ and $\pi_{\calBprime_X}^9$ as follows.
    Let $\calC$ and $\calC'$ be colourings drawn (independently) from $\pi_{\calB_X}^9$ and $\pi_{\calBprime_X}^9$, respectively.
    Let $\calB$ and $\calBprime$ be the two colourings of $\partial R_u$ such that $\calB(w)=\calB_X(w)$ and $\calBprime(w)=\calBprime_X(w)$ for $w\in \partial R_X\cap \partial R_u$, and $\calB(u)=\calC(u)$ and $\calBprime(u)=\calC'(u)$. Let $\Psi'$ be a coupling of $\pi^9_\calB$ and $\pi^9_\calBprime$ which we will define shortly. In a pair of colourings drawn from $\Psi$, the vertex $u$ is assigned the colours $\calC(u)$ and $\calC'(u)$, respectively, and the other vertices of $R_u$ are assigned colours according to $\Psi'$. It remains to define~$\Psi'$.

    Let $E=\{e_1,\dots,e_k\}\subseteq \calE R_u$ be the boundary edges incident to either $w_X$ or $u$ such that $e_1$ is incident to $w_X$, and $\{w_X,u\}\prec e_1\prec \cdots \prec e_5$ and $\{w_X,u\}\prec e_6\prec \cdots \prec e_k$.
    Note that five neighbours of $w_X$ are in $R_u$ and both $e_6$ and $e_k$ are edges between $u$ and neighbours of $w_X$. An example is illustrated in Figure~\ref{fig:vertex-u}.
    \begin{figure}[t]
        \centering
        \showgraphics{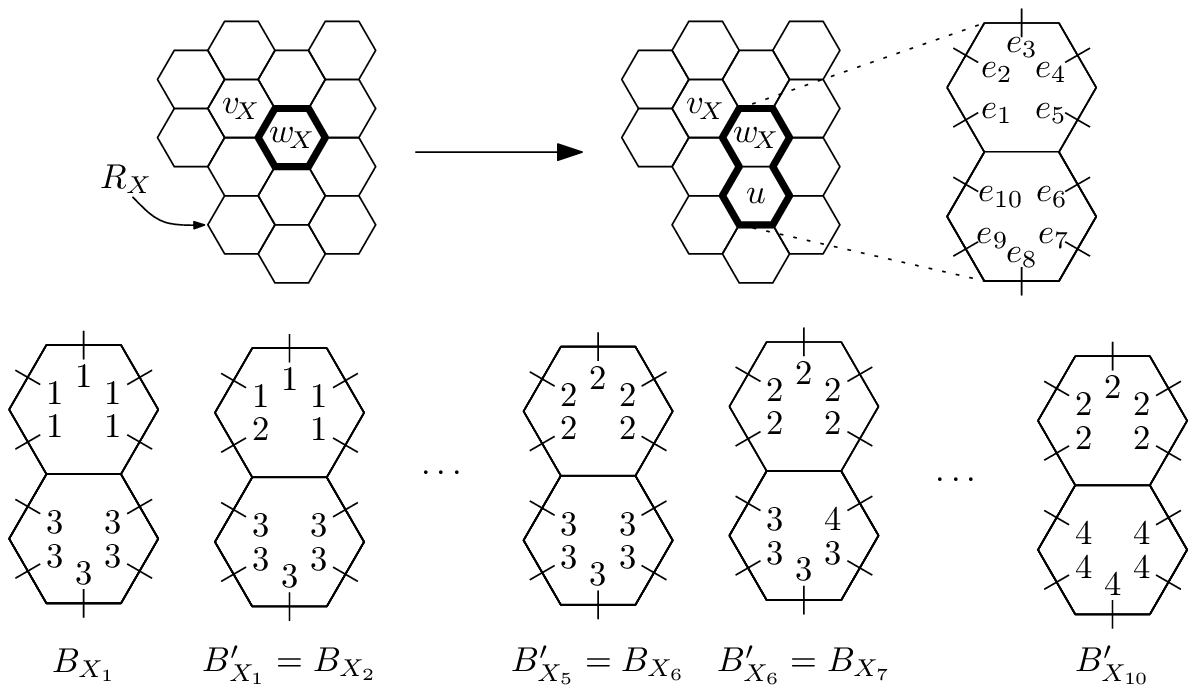}
        \caption{\label{fig:vertex-u}An example of edge boundary colourings of the edge-boundary pairs $X_1,\dots,X_{10}$ constructed from a vertex-boundary pair $X$ for which all six neighbours of $w_X$ are in $R_X$. Here we suppose that $\mathcal{B}_X(w_X)=1$ and $\mathcal{B}'_X(w_X)=2$, and $\calC(u)=3$ and $\calC'(u)=4$.}
    \end{figure}

    Similarly to above, we define $k$ edge-boundary pairs.
    For $i\in [k]$, let $X_i$ be the edge-boundary pair consisting of the region $R_{X_i}=R_u$, the distinguished edge $e_{X_i}= e_i$ and boundary colourings that differ on $e_{X_i}$.
    Note that $\pi^9_\calB = \pi^9_{B_{X_1}}$ and $\pi^9_\calBprime = \pi^9_{\Bprime_{X_k}}$.
    See Figure~\ref{fig:vertex-u} for an example.

    Let $\Psi_i$ be a coupling of
    $\pi_{B_{X_i}}^9$ and $\pi_{\Bprime_{X_i}}^9$ such that
    \begin{equation*}
        \label{eq:exp-decay-worse}
        \sum_{v\in R} \expect{1_{\Psi_i,v}} \leq
            \frac{5}{\epsilon} (1-\epsilon)^{d_{R_X}\!(w_X,R)-1}\,.
    \end{equation*}
    The $-1$ in the exponent comes from the fact that for some $X_i$, the edge $e_{X_i}$ is incident to $u$ and $d_{R_X}\!(u,R)$ might be $d_{R_X}\!(w_X,R)-1$ for some $R$.

    The coupling $\Psi'$ is defined by composing the couplings $\Psi_1,\dots,\Psi_k$. Thus, with $k$ being at most 10,
    \begin{equation*}
        \sum_{v\in R} \expect{1_{\Psi',v}} \leq
        \sum_{v\in R} \sum_{i=1}^k \expect{1_{\Psi_i,v}} =
        \sum_{i=1}^k \sum_{v\in R} \expect{1_{\Psi_i,v}} \leq
        \frac{50}{\epsilon} (1-\epsilon)^{d_{R_X}\!(w_X,R)-1}\,.
        \qedhere
    \end{equation*}
\end{proof}

The main theorem of the paper can now be proved.

\setcounter{tempcounter}{\value{theorem}}
\setcounter{theorem}{\value{ssmcounter}}
\begin{theorem}
    \ssmthm
\end{theorem}
\setcounter{theorem}{\value{tempcounter}}

\begin{proof}
    Let $R$ by any region, $R'\subseteq R$, $w\in \partial R$, and $\calB$ and $\calBprime$ two partial 9-colourings of $\partial R$ that are identical on all vertices except for $w$.
    Suppose that $\Psi$ is a coupling of $\pi_\calB^9$ and $\pi_\calBprime^9$ and let $(\calC, \calC')$ be a pair of colourings drawn from $\Psi$. Then
    \begin{equation}
        \label{eq:ssm-bound}
        \dtv(\pi_\calB^9(R'),\pi_\calBprime^9(R'))\leq
        \prob(\text{$\calC\neq \calC'$ on $R'$})\leq
        \sum_{v\in R'} \expect{1_{\Psi,v}}\,.
    \end{equation}

    Let $X$ be the vertex-boundary pair such that $R_X=R$, $w_X=w$, $\calB_X=\calB$ and $\calBprime_X=\calBprime$.
    Let $\epsilon=1/1000$, $\beta=\frac{50}{\epsilon(1-\epsilon)}$ and $\beta'=-\ln(1-\epsilon)$.
    Note that a coupling of $\pi_{\calB_X}^9$ and $\pi_{\calBprime_X}^9$ is also a coupling of $\pi_\calB^9$ and $\pi_\calBprime^9$. Using Lemma~\ref{lem:vertex-decay}, we know there is a coupling $\Psi'$ of $\pi_{\calB_X}^9$ and $\pi_{\calBprime_X}^9$ such that
    \begin{equation}
        \label{eq:ssm-exp}
        \sum_{v\in R'} \expect{1_{\Psi',v}} \leq
            \beta e^{-\beta' d_{R_X}(w_x,R')}\,.
    \end{equation}
    Let $\Psi=\Psi'$. Strong spatial mixing follows from Equations~(\ref{eq:ssm-bound}) and~(\ref{eq:ssm-exp}).
\end{proof}


\section{Rapidly mixing Glauber dynamics}
\label{sec:glauber}

We use Theorem~8 of Goldberg, Martin and Paterson~\cite{gmp-ssmfc-04} to show that the Glauber dynamics is rapidly mixing. Before applying their theorem we must introduce some notation.

Let $\ball_d(v)$ denote the set of vertices in $V_\trilat$ that are at distance at most $d$ from the vertex $v\in V_\trilat$. Thus we have $\ball_0(v)=\{v\}$. From the definition of the triangular lattice it follows that $|\partial \ball_d(v)|$, the number of vertices at distance $d+1$ from $v$, is of order $\Theta(d)$, hence $|\ball_d(v)|\in \Theta(d^2)$. One can show that $|\partial\ball_d(v)|=6(d+1)$ but we do not need to be that precise here (see Lemma~2.22 in~\cite{jalsenius-thesis-08} for a proof). It follows that for all~$v$, $|\partial\ball_d(v)|/|\ball_d(v)|\rightarrow 0$ as $d\rightarrow \infty$. In other words, uniformly in $v$, the ``surface-area-to-volume'' ratio of balls can be made arbitrarily small with a suitable choice of radius $d$.
This property of a graph is known as \emph{neighbourhood-amenability}.


Goldberg, Martin and Paterson~\cite{gmp-ssmfc-04} introduced the following definition of an \emph{$\epsilon$-coupling cover}.

\begin{definition}[Definition~4 of~\cite{gmp-ssmfc-04} quoted exactly]
    Let $G$ denote an infinite graph with maximum degree $\Delta$. Fix $\epsilon>0$. We say that $G$ has an \emph{$\epsilon$-coupling cover} if for all vertex-boundary pairs $X$, there is a coupling $\Psi_X$ of $\pi_{\calB_X^1}$ and $\pi_{\calB_X^2}$ such that
    \begin{equation*}
        \sum_{f\in R_X} \expect{1_{\Psi_X, f}} \leq \frac{\Delta}{\epsilon}\,.
    \end{equation*}
\end{definition}

We apply the definition to the triangular lattice. From Lemma~\ref{lem:vertex-decay} it follows that for every vertex-boundary pair $X$ there is a coupling $\Psi$ of $\pi^9_{\calB_X}$ and $\pi^9_{\calBprime_X}$ such that
\begin{equation}
    \sum_{v\in R_X} \expect{1_{\Psi,v}} \leq \frac{50}{1/1000\cdot(1-1/1000)}=\frac{\Delta}{\epsilon}\,,
\end{equation}
where $\Delta=6$ is the maximum degree of the triangular lattice and $\epsilon$ is the appropriate constant.
Thus, the triangular lattice has an $\epsilon$-coupling cover.

We have the following theorem from~\cite{gmp-ssmfc-04}.

\begin{theorem}[Theorem~8 of~\cite{gmp-ssmfc-04} quoted exactly]
    Let $G$ denote an infinite neighbourhood-amenable graph with maximum degree $\Delta$. Let $R$ be a finite subgraph of $G$ with $|R|=n$ and $\calB(R)$ denote a colouring of $\partial(R)$ using the colours $Q\cup \{0\}$. (We assume that $q\geq \Delta+2$.)

    Suppose there exists $\epsilon>0$ such that $G$ has an $\epsilon$-coupling cover. Then the Glauber dynamics Markov chain on $S(\calB(R))$ is rapidly mixing and $\tau(\delta)\in O(n(n+\log\frac{a}{\delta}))$.
\end{theorem}

Since the triangular lattice is neighbourhood-amenable and has an $\epsilon$-coupling cover, Theorem~8 of~\cite{gmp-ssmfc-04} translates directly to our Theorem~\ref{thm:rapid-mixing}, which is repeated below.

\setcounter{tempcounter}{\value{theorem}}
\setcounter{theorem}{\value{glaubercounter}}
\begin{theorem}
    \glauberthm
\end{theorem}
\setcounter{theorem}{\value{tempcounter}}


\section{Computations}
\label{sec:computations}

In this section we go through the computational steps involved in obtaining the 39 $\mu$-values from Table~\ref{tab:mu-values} and proving Lemma~\ref{lem:alphas}. The computations of the $\mu$-values are rather demanding and took around two weeks to run on a fairly powerful home PC of year 2006. We have used the language C for this task. For proving Lemma~\ref{lem:alphas}, we have used the language Python. Details on the implementation is given in Appendix~\ref{app:implementation}. The source code is available at \url{http://arxiv.org/abs/0706.0489}\,.


\subsection{Computing the $\mu$-values from Table~\ref{tab:mu-values}}
\label{sec:computing-mus}

Suppose that $X$ is an edge-boundary pair.
Let $\Omega$ be the set of colourings $\calC\in \Omega^9_{\Bprime_X}$ such that $\calC(v_X)= B_X(e_X)$, and let $\Omega'$ be the set of colourings $\calC\in \Omega^9_{B_X}$ such that $\calC(v_X)= \Bprime_X(e_X)$.
Let $\Omega_\text{both} = \Omega^9_{B_X}\cap \Omega^9_{\Bprime_X}$ be the set of proper 9-colourings of $R_X$ that agree with both $B_X$ and $\Bprime_X$.
\begin{lemma}
\label{lem:mu-program}
    For any edge-boundary pair X,
    \begin{equation*}
        \mu(X)= \frac{\max\{|\Omega|,|\Omega'|\}}{
            |\Omega_\textup{both}| + \max\{|\Omega|,|\Omega'|\}}\,.
    \end{equation*}
\end{lemma}
\begin{proof}
    Recall from the definition of $\mu(X)$ that $\mu(X)= \expect{1_{\Psi_X,v_X}}$.

    Suppose first that $|\Omega|\geq |\Omega'|$. In order to minimise $\expect{1_{\Psi,v_X}}$ we construct a coupling $\Psi$ of $\pi_{B_{X}}$ and $\pi_{\Bprime_{X}}$ as follows. Let $(\calC,\calC')$ be a pair of colourings drawn from $\Psi$. If $\calC'(v_X)= B_X(e_X)$ then $\calC(v_X)\neq \calC'(v_X)$ because $\calC$ is drawn from $\pi_{B_X}$, preventing $v_X$ from receiving the colour $B_X(e_X)$.
    However, if $\calC'(v_X)\neq B_X(e_X)$ then we choose  $\calC(v_X)= \calC'(v_X)$, which is always possible under the assumption that $|\Omega|\geq |\Omega'|$. Hence $\expect{1_{\Psi_X,v_X}}= |\Omega|/(|\Omega_\text{both}|+ |\Omega|)$.

    Suppose second that $|\Omega'|\geq |\Omega|$. By symmetry we have that $\expect{1_{\Psi_X,v_X}}=\nolinebreak[3] |\Omega'|/(|\Omega_\text{both}|+ |\Omega'|)$.
\end{proof}

For each of the 39 regions $M_i$ in Figure~\ref{fig:M-regions} we have written a program in C which computes $\mu_i$. We use the region $M_1$ to illustrate how $\mu_1$ is obtained. The other $\mu$-values are computed similarly.

Let $X$ be an edge-boundary pair such that $R_X=M_1$, $v_X=v_M$, $w_X=w_M$, $B_X$ and $\Bprime_X$ assign the colours $c_1,\dots,c_{18}$ to the boundary edges in $\calE R_X\setminus \{e_X\}$, $B_X(e_X)=1$ and $\Bprime_X(e_X)=2$. See Figure~\ref{fig:compute-m1}.
\begin{figure}[t]
    \centering
    \showgraphics{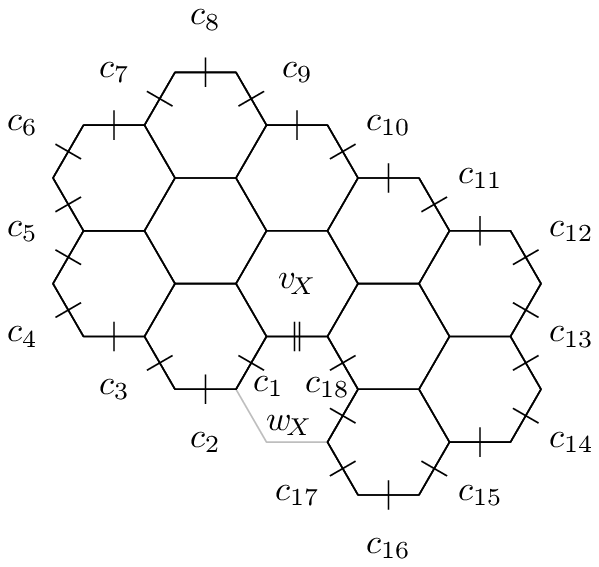}
    \caption{
        \label{fig:compute-m1}
        An edge-boundary pair $X$ with region $R_X=M_{1}$ and boundary edges coloured $c_1,\dots,c_{18}$.}
\end{figure}

For $i\in [9]$, let $n_i$ be the number of proper 9-colourings of $R_X$ that agree with the colouring $c_1,\dots,c_{18}$ of the boundary, \emph{disregarding} the colour of the edge $e_X$, and assign the colour $i$ to $v_X$. Thus, $n_1=|\Omega|$, $n_2=|\Omega'|$ and $n_3+\cdots +n_9= |\Omega_\text{both}|$.
We write a subroutine that computes the values $n_i$ given the colours $c_1,\dots,c_{18}$. Computing them in a brute force manner will take too long so we must be a little more clever than that. We construct a dynamic programming table which lets us reuse the number of colourings computed for subsets of the region $R_X$. For details, see Appendix~\ref{app:implementation}.

Let $m=n_1/(n_1+n_3+\cdots + n_9)$ be a function of $c_1,\dots,c_{18}$. We loop through the colours $c_1,\dots,c_{18}$, and for each configuration we compute $m$. It follows that $\mu_1$ is the largest value of $m$ that we encounter.

Each colour $c_j$ can take a value from the set $\{0,1,\dots,9\}$. Looping through all $10^{18}$ configurations of $c_1,\dots,c_{18}$ yields an unnecessary large number of redundant boundary colourings. Instead of considering all $10^{18}$ configurations, we keep the number down by making a few useful observations:

\begin{itemize}
\item Swapping the colours $c_{16}$ and $c_{17}$ will not change the value $m$. Hence we may skip colourings for which $c_{16}>c_{17}$.

\item The colours 0, 1 and~2 have a special meaning here since 0 symbolises ``no colour'' and 1 and~2 are used on the edge $e_X$. However, the colours $3,\dots,9$ are merely labels and therefore there is no reason to use a colour $c'\in \{4,\dots,9\}$ for $c_j$ unless the colour $c'-1$ has been used for some $c_{j'}$, where $j'<j$.

\item From the definition of an edge-boundary pair, $c_1, c_{18}\in \{1,2\}$ and $c_1\neq c_{18}$.
\end{itemize}

Obtaining the 39 $\mu$-values, using the observations described above, took around two weeks on a fairly powerful home PC of year 2006. We left a computer running non-stop for 24 hours per day without using it for other purposes.

\subsection{The experimental phase}
\label{sec:experimental}

A very reasonable question to ask is why we used exactly those 39 regions $M_i$ in Figure~\ref{fig:M-regions} and the regions $G$ and $F$. The regions are the result of a rather long experimental phase where we started with a set of smaller regions and carried out the computations as described in this article. Initially, we failed to prove Lemma~\ref{lem:alphas}. That is, we were not able to find constants $\alpha_i$ such that the equation in Lemma~\ref{lem:alphas} would hold for all regions in $\calG$. The reason for this is twofold: too small sizes of $G$ and $F$ do not allow enough recursions, and too small regions $M_i$ yield too large values $\mu_i$. Gradually we increased the sizes of the regions until Lemma~\ref{lem:alphas} could successfully be proved. In order to check whether the equations in the statement of Lemma~\ref{lem:alphas} could all be satisfied, we used the free linear program solver GLPK (GNU Linear Programming Kit). The constants found by the solver are the constants we use in the implementation of the proof of Lemma~\ref{lem:alphas}, which is described in detail in Appendix~\ref{app:implementation}. It should also be mentioned that the choice of $\epsilon$ played a role. For instance, we would have failed solving the linear program if we had used $\epsilon=1/100$ instead of the smaller $\epsilon=1/1000$.

While increasing the sizes of $G$ and $F$, we also computed $\mu$-values for growing regions $M_i$. In total we considered a few hundred distinct regions $M_i$, of which some were even larger than the region $M_1$ in Figure~\ref{fig:M-regions}. One might ask how we managed to compute the $\mu$-values for such a vast number of regions given that it took two weeks of computations for the 39 regions in Figure~\ref{fig:M-regions}. Instead of computing the $\mu$-values exactly, we used a hill climbing technique where we randomised colourings of the boundary, to which we iteratively made small changes, whereby larger and larger values $m$ (see previous section) were found. This process allowed us to build a ``library'' of randomised $\mu$-values. In practice, we let a computer run during the night over some time in order to obtain hopefully good estimates of the $\mu$-values. Interestingly, it turned out that for many regions, only a few minutes running time was enough to yield a value of $m$ that did not seem to increase further. For such regions we stopped the hill climbing process after a couple of hours.

Once Lemma~\ref{lem:alphas} had been successfully proved with randomised $\mu$-value estimates, we were faced with the task of computing the exact $\mu$-values. Initially the set of $M$-regions in the proof was rather large, so first we pruned the set by carefully choosing regions to throw away. Eventually we ended up with the 39 regions in Figure~\ref{fig:M-regions} for which the $\mu$-values were computed exactly. It is interesting to note that the 39 randomised $\mu$-value estimates were identical to the exact values.

The successful use of the $\mu$-value estimates suggests that our approach could be used to prove better mixing bounds for other lattices. Although the system of inequalities that was solved in the proof of Lemma~\ref{lem:alphas} contained a huge number of inequalities (around 100,000), the real bottleneck seemed to be the demanding computations of the $\mu$-values.

\section{Acknowledgements}

The author would like to thank Leslie Ann Goldberg for helpful discussions, and the University of Liverpool where most of the work on this paper has been conducted.

\bibliographystyle{plain}
\bibliography{triangular}

\begin{thebibliography}{10}

\bibitem{ammb-sgcfc-04}
D.~Achlioptas, M.~Molloy, C.~Moore, and F.~{Van~Bussel}.
\newblock Sampling grid colorings with fewer colours.
\newblock In {\em LATIN 2004: Theoretical Informatics}, volume 2976 of {\em
  Lecture Notes in Computer Science}, pages 80--89. Springer, 2004.

\bibitem{bd-pctprmmc-97}
R.~Bubley and M.~Dyer.
\newblock Path coupling: a technique for proving rapid mixing in {M}arkov
  chains.
\newblock In {\em FOCS '97: Proceedings of the 38th Symposium on Foundations of
  Computer Science}, pages 223--231. IEEE Computer Society Press, 1997.

\bibitem{bdgj-accsdg-99}
R.~Bubley, M.~Dyer, C.~Greenhill, and M.~Jerrum.
\newblock On approximately counting colourings of small degree graphs.
\newblock {\em SIAM Journal on Computing}, 29(2):387--400, 1999.

\bibitem{dgj-dcss-06}
M.~Dyer, L.~A. Goldberg, and M.~Jerrum.
\newblock Dobrushin conditions and systematic scan.
\newblock In {\em Approximation, Randomization, and Combinatorial Optimization.
  Algorithms and Techniques}, volume 4110 of {\em Lecture Notes in Computer
  Science}, pages 327--338. Springer, 2006.

\bibitem{dsvw-mtslss-04}
M.~Dyer, A.~Sinclair, E.~Vigoda, and D.~Weitz.
\newblock Mixing in time and space for lattice spin systems: a combinatorial
  view.
\newblock {\em Random Structures and Algorithms}, 24(4):461--479, 2004.

\bibitem{g-gmpt-88}
H.-O. Georgii.
\newblock {\em Gibbs measures and phase transitions}.
\newblock de Gruyter Studies in Mathematics 9. Walter de Gruyter \& Co.,
  Berlin, Germany, 1988.

\bibitem{ghm-rgef-01}
H.-O. Georgii, O.~H\"{a}ggstr\"{o}m, and C.~Maes.
\newblock The random geometry of equilibrium phases.
\newblock {\em Phase Transitions and Critical Phenomena}, 18:1--142, 2001.

\bibitem{gjmp-imbafpm-06}
L.~A. Goldberg, M.~Jalsenius, R.~Martin, and M.~Paterson.
\newblock Improved mixing bounds for the anti-ferromagnetic potts model on
  $\mathbb{Z}^2$.
\newblock {\em LMS Journal of Computation and Mathematics}, 9:1--20, 2006.

\bibitem{gmp-ssmfc-04}
L.~A. Goldberg, R.~Martin, and M.~Paterson.
\newblock Strong spatial mixing with fewer colours for lattice graphs.
\newblock {\em SIAM Journal on Computing}, 35(2):486--517, 2005.

\bibitem{jalsenius-thesis-08}
M.~Jalsenius.
\newblock {\em Spatial and Rapid Mixing on Lattice Graphs}.
\newblock PhD thesis, University of Liverpool, United Kingdom, 2008.

\bibitem{Jal09:kagome}
M.~Jalsenius.
\newblock Strong spatial mixing and rapid mixing with five colours for the
  kagome lattice.
\newblock {\em LMS Journal of Computation and Mathematics}, 12:195--227, 2009.

\bibitem{JK08:grid-scan}
M.~Jalsenius and K.~Pedersen.
\newblock A systematic scan for 7-colourings of the grid.
\newblock {\em International Journal of Foundations of Computer Science},
  19:1461–--1477, 2008.

\bibitem{j-vsa-95}
M.~Jerrum.
\newblock A very simple algorithm for estimating the number of $k$-colorings of
  a low-degree graph.
\newblock {\em Random Structures and Algorithms}, 7(2):157--165, 1995.

\bibitem{j-csi-03}
M.~Jerrum.
\newblock {\em Counting, Sampling and Integrating: Algorithms and Complexity}.
\newblock Birkh{\"{a}}user, Basel, Switzerland, 2003.

\bibitem{jvv-rgcs-86}
M.~Jerrum, L.~Valiant, and V.~Vazirani.
\newblock Random generation of combinatorial structures from a uniform
  distribution.
\newblock {\em Theoretical Computer Science}, 43:169--188, 1986.

\bibitem{m-lgddsm-97}
F.~Martinelli.
\newblock Lectures on {G}lauber dynamics for discrete spin models.
\newblock In {\em Lectures on Probability Theory and Statistics (Saint-Flour,
  1997)}, volume 1717 of {\em Lecture Notes in Mathematics}, pages 93--191.
  Springer, 1999.

\bibitem{m-vrmmc2dc-01}
M.~Molloy.
\newblock Very rapidly mixing {M}arkov chains for 2{$\Delta$}-colourings and
  for independent sets in a 4-regular graph.
\newblock {\em Random Structures and Algorithms}, 18(2):101--115, 2001.

\bibitem{p-dcssbd-07}
K.~Pedersen.
\newblock Dobrushin conditions for systematic scan with block dynamics.
\newblock In {\em Mathematical Foundations of Computer Science 2007}, volume
  4708 of {\em Lecture Notes in Computer Science}, pages 264--275. Springer,
  2007.

\bibitem{ss-apt-97}
J.~Salas and A.~Sokal.
\newblock Absence of phase transition for antiferromagnetic {P}otts models via
  the {D}obrushin uniqueness theorem.
\newblock {\em Journal of Statistical Physics}, 86(3--4):551, 1997.

\bibitem{v-ibsc-00}
E.~Vigoda.
\newblock Improved bounds for sampling colourings.
\newblock {\em Journal of Mathematical Physics}, 41(3):1555--1569, 2000.

\bibitem{w-mtsdss-04}
D.~Weitz.
\newblock {\em Mixing in Time and Space for Discrete Spin Systems}.
\newblock PhD thesis, University of California, Berkley, 2004.

\bibitem{w-ccugm-05}
D.~Weitz.
\newblock Combinatorial criteria for uniqueness of {G}ibbs measures.
\newblock {\em Random Structures and Algorithms}, 27(4):445--475, 2005.

\end{thebibliography}

\newpage


\appendix
\section{Implementation}
\label{app:implementation}

In this appendix, we describe the implementation of the programs that have been used in the computations. In Section~\ref{sec:implement-verify} we explain the program that proves Lemma~\ref{lem:alphas}, and in Section~\ref{sec:implement-mu} we explain the programs that compute the 39 $\mu$-values and thereby prove Lemma~\ref{lem:mu-values}. All programs are available at
\begin{equation*}
    \text{\url{http://arxiv.org/abs/0706.0489}}
\end{equation*}

\subsection{Proving Lemma~\ref{lem:alphas}}
\label{sec:implement-verify}

The program that proves Lemma~\ref{lem:alphas} is written in Python~2.6 and is called \texttt{lemma7.py}. There are numbered comments in the code, which are detailed below.

\subsubsection*{Comment 1}
The Fraction data type is imported to guarantee exact computations with rational numbers. \texttt{Fraction(x,~y)} represents the number $\frac{x}{y}$.

\subsubsection*{Comment 2}
The file \texttt{constants.py} is executed. It reads in the values of the constants $\alpha_1,\dots,\alpha_{2048}$. Here are the first few lines of the file \texttt{constants.py}.
\begin{verbatim}
alphaValues = {}
alphaValues["alpha10010010100"] = Fraction(2, 1)
alphaValues["alpha00000011001"] = Fraction(2, 1)
alphaValues["alpha01011111001"] = Fraction(20279, 10000)
alphaValues["alpha10101001111"] = Fraction(113631, 50000)
\end{verbatim}
\texttt{alphaValues} is a Dictionary data type in which each $\alpha$-variable, represented as a string, is associated with its value (a rational number). The variable name consists of a bit-string that encodes the $F$-region associated with the variable in the following way. Each vertex of the region $F$ is given a label from the set $\{0,\dots,10\}$ according to Figure~\ref{fig:codeF}. For $i\in\{0,\dots,10\}$, the $i$th bit (starting from the left) of the variable name is 1 if and only if the vertex labelled $i$ is in the region associated with the variable. In other words, the bit-string is a characteristic vector of the set of vertices of $F$.

\subsubsection*{Comment 3}
The program goes through all $\alpha$-variables, confirming that their values are in the range $[2,6]$.

\subsubsection*{Comment 4}
Defines $\epsilon=1/1000$.

\subsubsection*{Comment 5}
The 39 $\mu$-values are set. \texttt{mu[i]} holds the value $\mu_i$.

\subsubsection*{Comment 6}
The array \texttt{G} represents a subregion of $G$. Each vertex of the region $G$ is given a label from the set $\{0,\dots,16\}$ according to Figure~\ref{fig:codeG}. For $i\in\{0,\dots,16\}$, if \texttt{G[i]} is 1 then the vertex labelled $i$ is in the region represented by \texttt{G}. If \texttt{G[i]} is 0 then the vertex is not in the region.
\begin{figure}[t]
    \begin{minipage}[t]{0.47\linewidth}
        \centering
        \showgraphics{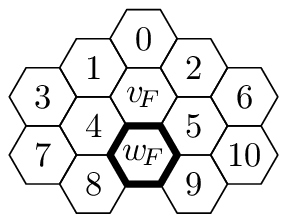}
        \caption{\label{fig:codeF}The region $F$ with labelled vertices.}
    \end{minipage}
    \hfill
    \begin{minipage}[t]{0.47\linewidth}
        \centering
        \showgraphics{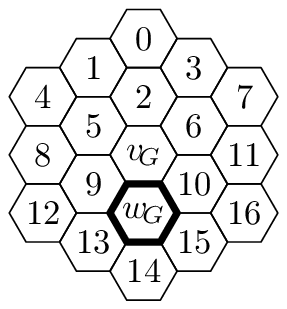}
        \caption{\label{fig:codeG}The region $G$ with labelled vertices.}
    \end{minipage}
\end{figure}

Starting with \texttt{G} containing all zeros, for each round in the while-loop, the array \texttt{G} is updated like a 17-digit binary counter. Thus, when \texttt{G[17]} is set to 1, all $2^{17}$ combinations of the first 17 bits have been considered, and the program breaks out from the while-loop (which ends the program). A configuration of the array \texttt{G} corresponds to a region $G'\in\calG$ in the statement of Lemma~\ref{lem:alphas}.

\subsubsection*{Comment 7}
If all five neighbours of $w_G$ (labels 9, 10, 13, 14, 15) are in \texttt{G}, proceed with the next region.

\subsubsection*{Comment 8}
Starting with an empty list \texttt{muList}, the program goes through all $i\in[39]$ and checks whether $M_i$ is a subregion of \texttt{G}. If so, $i$ is added to \texttt{muList}. The program also considers the mirrored version of $M_i$. Figure~\ref{fig:codeM-regions} illustrates the overlaps of the $M$-regions and $G$.

\subsubsection*{Comment 9}
The variable \texttt{minMu} is set to the $i$ such that $\mu_i$ is minimised over the $i$s in \texttt{muList}. Note that $\mu_{39}$ is the largest of all $\mu$-values.

\subsubsection*{Comment 10}
The regions \texttt{alphaA} through \texttt{alphaF} are subregions of $F$, defined according the overlaps of $F$ and $G$ in Figures~\ref{fig:codeIntersect-a}--\ref{fig:codeIntersect-f}, respectively. This corresponds to the overlaps in Figure~\ref{fig:intersections}. The strings \texttt{alphaAstring} through \texttt{alphaFstring} are the names of the $\alpha$-variables associated with each of the seven subregions of $F$.
\begin{figure}[t]
  \centering
  ~
  \hfill
  \subfloat[][]{\label{fig:codeIntersect-a}\showgraphics{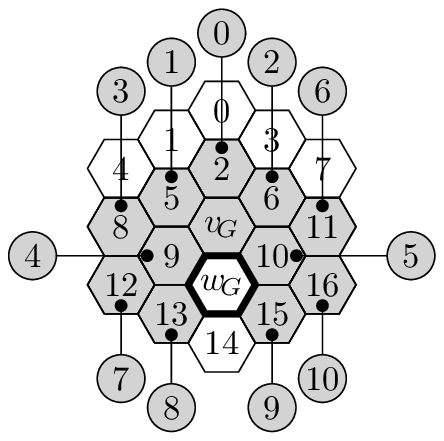}}
  \hfill
  \subfloat[][]{\label{fig:codeIntersect-b}\showgraphics{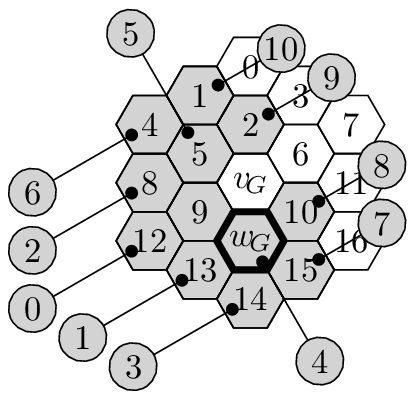}}
  \hfill
  \subfloat[][]{\label{fig:codeIntersect-c}\showgraphics{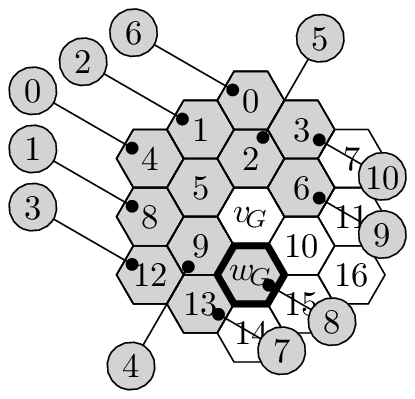}}
  \hfill
  ~

  ~
  \hfill
  \subfloat[][]{\label{fig:codeIntersect-d}\showgraphics{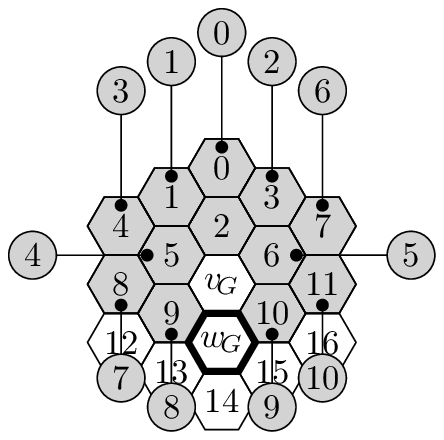}}
  \hfill
  \subfloat[][]{\label{fig:codeIntersect-e}\showgraphics{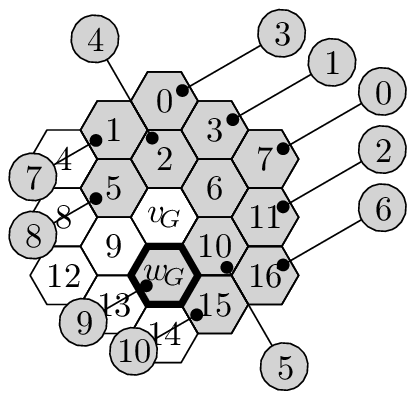}}
  \hfill
  \subfloat[][]{\label{fig:codeIntersect-f}\showgraphics{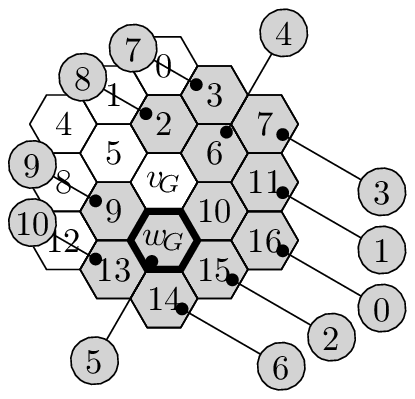}}
  \hfill
  ~
  \caption{Intersections of region $F$ and $G$, using the labelling from Figures~\ref{fig:codeF} and~\ref{fig:codeG}.}
  \label{fig:codeIntersections}
\end{figure}

\subsubsection*{Comment 11}
The program now verifies Equation~(\ref{eq:inequality}) in the statement of Lemma~\ref{lem:alphas}. The variables \texttt{LHS} and \texttt{RHS} correspond to the left hand side and right hand side, respectively, of the equation. First, for each neighbour of $v_F$ that is in the region $G'$ (\texttt{G} in the code), add the corresponding $\alpha$-value to \texttt{LHS}, and then finally multiply the sum with $\mu_m$ (which is \texttt{mu[minMu]} in the code).
The value of \texttt{RHS} is the value of \texttt{alphaAstring} times $(1-\epsilon)$.

\subsubsection*{Summary}
When running the program \texttt{lemma7.py}, we note that no ``complaints'' are outputted, which implies that Lemma~\ref{lem:alphas} is successfully proved.

\subsection{Computing the $\mu$-values}
\label{sec:implement-mu}

The programs that compute the values $\mu_1,\dots,\mu_{39}$ are called \texttt{mu1.c}, \texttt{mu2.c}, and so on, up to \texttt{mu39.c}, respectively. They are written in C. Using Python for this task would be way too slow. All 39 programs include the file \texttt{mutop.c}, which contains code that runs at the very beginning of every \texttt{mu}-file. The file \texttt{mutop.c} contains macro definition of various for-loops and if-statements, and also defines a large set of variables. The purpose of this file is to keep the code of the \texttt{mu}-files short and concise, with focus on the actual region and its vertices without cluttered C-syntax. This is important when checking the correctness of the programs.

In order to describe the programs, we use the file \texttt{mu10.c} as an example and go through its code in detail. The structure of the other programs is the same.

We break the code into blocks and describe each block separately.

\subsubsection*{The file \texttt{mu10.c}}

\verbatimStart{1}
/*
mu.0.1.2.3.5.6.7.8
31648
123341
Output from the program:
63296
246682
13h34m (exact)
\end{Verbatim}

The lines above are comments in the code. We focus on the relevant lines. Recall that $\mu_{10}=31648/123341$. Line~3 is the numerator and Line~4 is the denominator. The output from the program is not the $\mu$-value in its shortest form. The output is specified by Lines~6 and~7 as the numerator and denominator, respectively. Line~8 gives the time it took to run the program in hours and minutes. Obviously this varies from machine to machine, but it gives a rough estimate. The word ``exact'' here is referring to the computations as being exact and not randomised (remains from the past).

The next few lines of the code are given in Figure~\ref{fig:codeMuExample}. They illustrate a labelling of the vertices of the region $M_{10}$ and a labelling of the boundary vertices. To make the description more comprehensible, we also show a drawing of the region on top of the code.
\begin{figure}[t]
    \centering
    \showgraphics{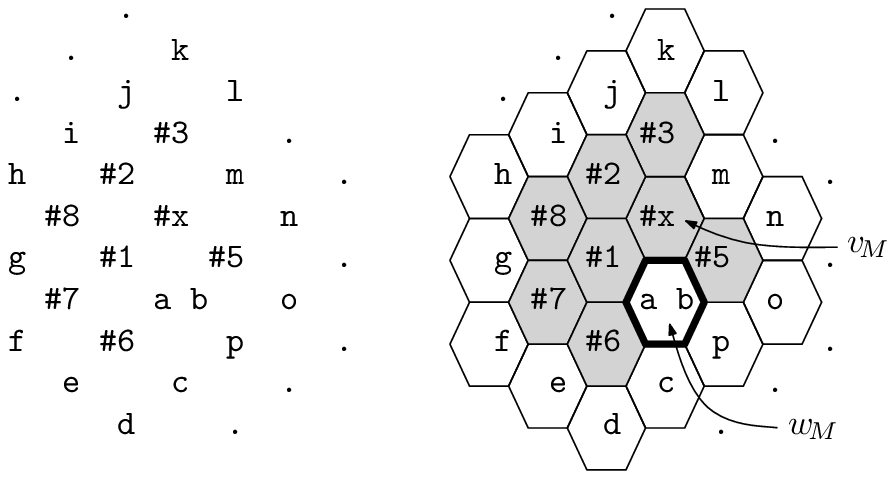}
    \caption{\label{fig:codeMuExample}The labelling of vertices in the code of \texttt{mu10.c}.}
\end{figure}

In the code, the numbering of the vertices of the region will be used in the variable names relating to them. The variables $\texttt{a},\texttt{b},\dots,\texttt{p}$ represent the boundary vertices. They define the colours of the boundary edges. For this reason, \texttt{a} and \texttt{b} are two variables of the same vertex $w_M$; \texttt{a} is the colour of the boundary edge incident to \texttt{\#1} and \texttt{\#6}, and \texttt{b} is the colour of the boundary edge incident to \texttt{\#5}.

\verbatimStart{9}
#include "mutop.c"
BEGIN;
AB C D_ E F G H I J K L_ M N O_ P_
{
\end{Verbatim}

Line~9 includes the file \texttt{mutop.c}, which defines macros and initiates the code with the \texttt{main()}-body, in which relevant variables are declared. We use upper-case letters for macros and constants, and lower-case letters for variables.

Line~10 records the start time of the execution. This line is irrelevant for the actual computation of $\mu$.

Every letter of Line~11 corresponds to a for-loop. That is, Line~11 defines a list of nested for-loop, each of which corresponds to the colour of a boundary edge. For example, \texttt{AB} loops through the two configurations $\texttt{a}=1,\texttt{b}=2$ and $\texttt{a}=2,\texttt{b}=1$, which are two choices of \texttt{a} and \texttt{b} we need to consider.

The macro \texttt{C} loops trough possible colours of \texttt{c}, which are 0, 1, 2 and 3. As described earlier, we do not need to consider other colours than those on \texttt{c}. The macro definitions take care of the bounds on possible values of the variables. For this particular macro \texttt{C}, the following lines from \texttt{mutop.c} are relevant.

\begin{verbatim}
#define FOR_START(c, uc)  for(c = QL, uc = 3; c <= uc; c++)
#define C                 FOR_START(c, uc)
\end{verbatim}

\texttt{QL} is 0 and is the smallest colour and \texttt{uc} is the upper bound on colours of \texttt{c}, which is always~3. As we will see below, the variables representing such upper bounds vary as we progress though the for-loops.

Next on Line~11 is the macro \texttt{D\_}. The underscore character specifies that \texttt{d} it is tied closely to the previous colour \texttt{c} in the sense that the ordering of \texttt{c} and \texttt{d} does not matter when computing the $\mu$-value. This symmetry was explained in Section~\ref{sec:computing-mus}. The for-loop associated with \texttt{d} ensures that \texttt{d} is never smaller than \texttt{c}. Here are the relevant lines from \texttt{mutop.c}.

\begin{verbatim}
#define FOR_NEXT_(d, ud, c, uc)
        for(d = c,  ud = min(max(uc, c+1), 9); d <= ud; d++)

#define D_  FOR_NEXT_(d, ud, c, uc)
\end{verbatim}

We see above that the upper bound \texttt{ud} on values of \texttt{d} depends on the current value of \texttt{c}; only if \texttt{c} has reached the value 3, we allow \texttt{d} to take on the next value~3+1=4. Otherwise the maximum value of \texttt{d} is set by the upper bound on \texttt{c}, which is \texttt{uc} (which actually is~3).

The remaining part of Line~11 defines the other for-loops associated with the other boundary vertices. Line~12 specifies the opening of the body of the innermost for-loop.

\verbatimStart{13}
  FOR_VERTEX(x)
  {
    INITX1(m);
\end{Verbatim}

The variable \texttt{x} represents the colour of the vertex $v_M$. Line~13 is a for-loop that takes the \texttt{x} through the values $1,\dots,9$. Recall from Section~\ref{sec:computing-mus} that for $i\in [9]$, $n_i$ is the number of proper 9-colourings of the region that agree with the colouring of the boundary and assign the colour $i$ to $v_M$. The code in the body of the for-loop of Line~13 computes $n_\texttt{x}$. The values are stored in the array \texttt{vx}, where \texttt{vx[x]} is $n_\texttt{x}$.

Since we do not need the value of $n_2$ when computing $\mu$, we skip $\texttt{x}=2$. The macro \texttt{INITX1} on Line~15 skips to next value of \texttt{x} if $\texttt{x}=2$, otherwise it sets \texttt{vx[x]} to~0. Further, it tests if the value of \texttt{x} is the same as the value of \texttt{m} (the boundary vertex adjacent to $v_M$). If this is the case, we will have a colouring that does not agree with the boundary and we can skip immediately to the next value of \texttt{x}.

The number~1 in the macro name \texttt{INITX1} refers to the fact that one vertex has to be checked against \texttt{x}, in this case \texttt{m}. The macro \texttt{INITX2}, which is used for other $M$-regions, takes two arguments, and so on. Many macro names contain a number. The number is referring to the number of relevant arguments that the macro takes.

\verbatimStart{16}
    FOR_VERTEX(m1)
    {
      SKIP2(m1, x, a);
\end{Verbatim}

The variable \texttt{m1} is the colour of the vertex labelled \texttt{\#1} in Figure~\ref{fig:codeMuExample}. The for-loop on Line~16 takes \texttt{m1} from 1 to 9. In order to avoid non-proper colourings or colourings that do not agree with the boundary, we have to skip values of \texttt{m1} for which $\texttt{m1}=\texttt{x}$ or $\texttt{m1}=\texttt{a}$. This test is done with the macro on Line~18.

Unlike vertex \texttt{\#1}, we do not use for-loops for the other vertices \texttt{\#2}, \texttt{\#3}, \texttt{\#5}, \texttt{\#6}, \texttt{\#7} and \texttt{\#8}. The running time would be too long with that many nested for-loops. Instead we resort to a dynamic programming approach. Similarly to the array \texttt{vx}, we use one array per vertex. The arrays are called \texttt{v2}, \texttt{v3}, \texttt{v5}, \texttt{v6}, \texttt{v7} and \texttt{v8}, respectively. The ordering of the arrays by which we fill in their values is first \texttt{v6}, then \texttt{v7}, \texttt{v8}, \texttt{v2}, \texttt{v3}, and lastly \texttt{v5}. For $i\in [9]$, the value of, for instance $\texttt{v2[}i\texttt{]}$, is the number of proper colourings of the subregion consisting of the vertices up to and including \texttt{\#2}, which are \texttt{\#6}, \texttt{\#7}, and \texttt{\#8}, subject to the boundary and the value of \texttt{x} and \texttt{m1}, and such that vertex \texttt{\#2} has colour $i$.

\verbatimStart{19}
      ONES (v6);        ZERO5 (v6, m1, a, c, d, e);
      NEXT (v7, v6);    ZERO4 (v7, m1, e, f, g);
      NEXT (v8, v7);    ZERO4 (v8, i, m1, g, h);
      NEXT (v2, v8);    ZERO4 (v2, j, x, m1, i);
      NEXT (v3, v2);    ZERO5 (v3, k, l, m, x, j);
      GAP  (v5, v3);    ZERO6 (v5, m, n, o, p, b, x);
\end{Verbatim}

Let us start with Line~19. The array \texttt{v6} is first initiated with 1 at every position. The following macro, \texttt{ZERO5}, sets the positions \texttt{m1}, \texttt{a}, \texttt{c}, \texttt{d} and \texttt{e} of \texttt{v6} to~0. The reason is that vertex \texttt{\#6} cannot have the colour specified by these variables as this would violate the property of the colouring being proper or in agreement with the boundary. For the valid colours of \texttt{\#6}, there is only one colouring, since the subregion consists only of the single vertex~\texttt{\#6}.

On Line~20, \texttt{NEXT(v7, v6)} fills in the values of the array \texttt{v7}, subject to the colour of \texttt{v6}. For $j\in[9]$,
\begin{equation}
    \label{eq:codeMuSum}
    \texttt{v7[}j\texttt{]} = \left(\sum_{i=1}^9  \texttt{v6[}i\texttt{]}\right) - \texttt{v6[}j\texttt{]}\,.
\end{equation}
The number of colourings of the subregion $\{\texttt{\#6, \texttt{\#7}}\}$ such that \texttt{\#7} has colour~$j$ is obtained by summing the colourings of the subregion $\{\texttt{\#6}\}$, excluding the case when \texttt{\#6} has colour $j$ as this would not make the colouring proper. Hence the subtraction of \texttt{v6[}j\texttt{]}. In the code of \texttt{mutop.c}, there is a macro called \texttt{SUM} which calculates the sum. It is called from the \texttt{ZERO}-macros to prepare for the subsequent \texttt{NEXT}-macro. Once the \texttt{NEXT}-macro has been executed, we must make sure that colourings not in agreement with the boundary are excluded. For instance, \texttt{v7[3]} must be~0 if the boundary edge specified by \texttt{g} is 3. The macro \texttt{ZERO4} on Line~20 takes care of this, setting appropriate elements of \texttt{v7} to zero according to the indices specified by \texttt{m1}, \texttt{e}, \texttt{f} and \texttt{g}.

The program proceeds with the remaining vertices on Lines~21--24. The difference between the macro \texttt{NEXT} and the macro \texttt{GAP} on Line~24 is that the two vertices referred to in the argument of \texttt{GAP}, here \texttt{\#3} and \texttt{\#5}, are not adjacent. In this case, the array \texttt{v5} is filled in according to Equation~(\ref{eq:codeMuSum}) but without the subtraction.

\verbatimStart{25}
      UPDATEX(v5);
    }
  }
  UPDATEMU;
}
END;
return(0);
}
\end{Verbatim}

After Line~24, the sum $\sum_{i=1}^9  \texttt{v5[}i\texttt{]}$ is the number of proper colourings of the whole region such that vertex $v_M$ has colour \texttt{x} and vertex \texttt{\#1} has colour \texttt{m1}. The macro \texttt{UPDATEX(v5)} on Line~25, inside the body of the for-loop of \texttt{m1}, computes the sum and adds it to \texttt{vx[x]}. Thus, over all values of \texttt{m1}, the final value of \texttt{vx[x]} is computed. The closing bracket on Line~27 indicates the end of the body of the for-loop of \texttt{x}. When Line~28 is reached, the vector \texttt{vx} contains every value we need.

\texttt{UPDATEMU} on Line~28 computes the $\mu$-value
\begin{equation*}
    \mu = \frac{\texttt{vx[1]}}{\texttt{vx[1]}+\left(\sum_{i=3}^9 \texttt{vx[}i\texttt{]}\right)}
\end{equation*}
and compares it to the maximum of all $\mu$-values computed so far. The maximum value of $\mu$ is stored with the two variables \texttt{maxnum} and \texttt{maxden}, where \texttt{maxnum} is the numerator and \texttt{maxden} is the denominator. Both variables are integers of the data type double to allow enough digits. The comparison of $\mu$ with $\texttt{maxnum}/\texttt{maxden}$ is performed by multiplying with the denominators to ensure integer arithmetic. If the newly computed value of $\mu$ exceeds the previously largest $\mu$-value, \texttt{maxnum} and \texttt{maxden} are updated accordingly.

On Line~29, the program proceeds to the next boundary and computes the $\mu$-value over again.

Finally, on Line~30, after all boundaries have been considered, the largest $\mu$-value is outputted as well as how long the program has been running.

\begin{figure}[p]
    \centering
    \showgraphics{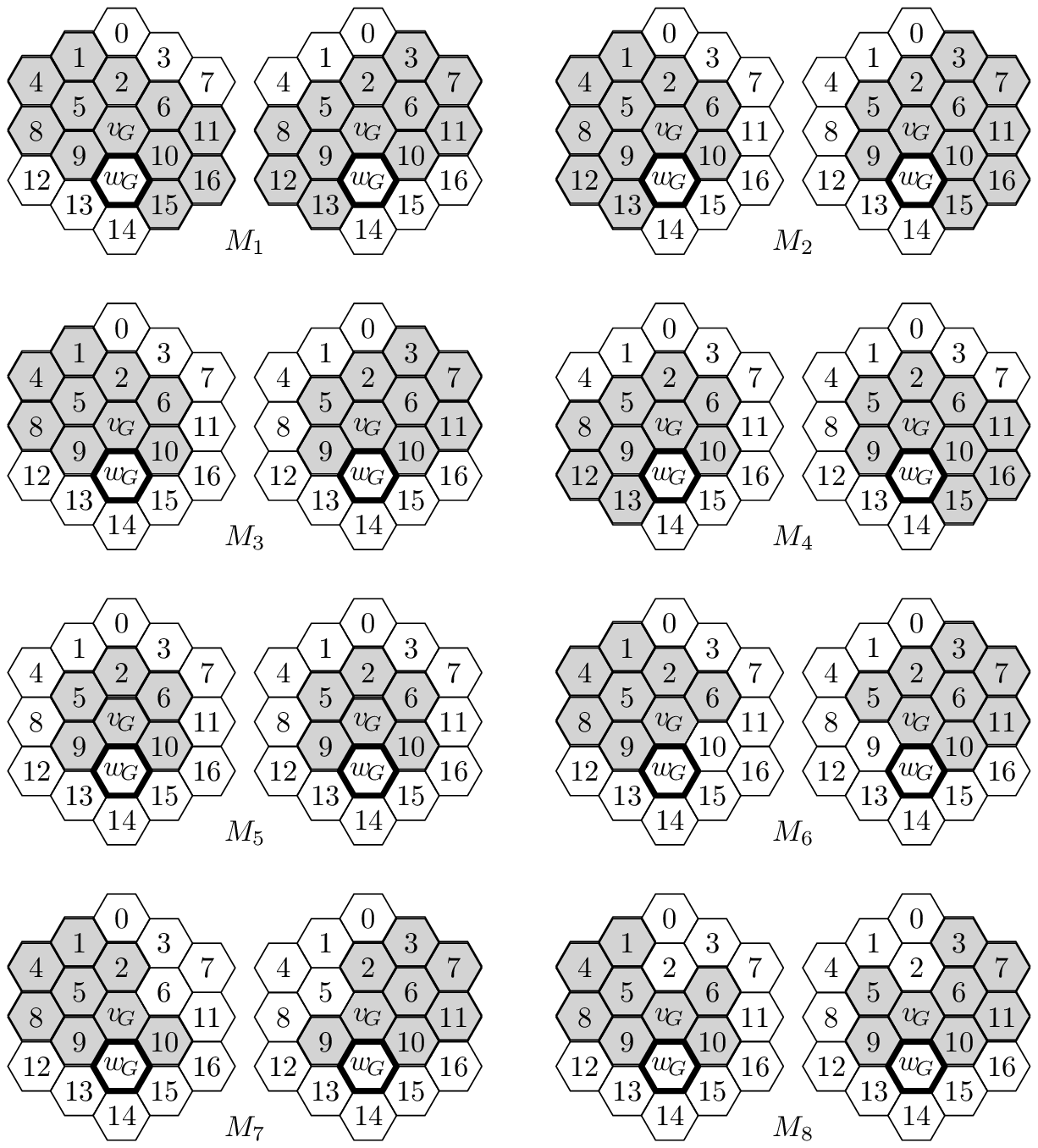}
    \bigskip

    Part 1 of 5
    \caption{\label{fig:codeM-regions}
        The regions $M_1,\dots,M_{39}$ in labelled region~$G$.}
\end{figure}
\begin{figure}[p]
    \addtocounter{figure}{-1}
    \centering
    \showgraphics{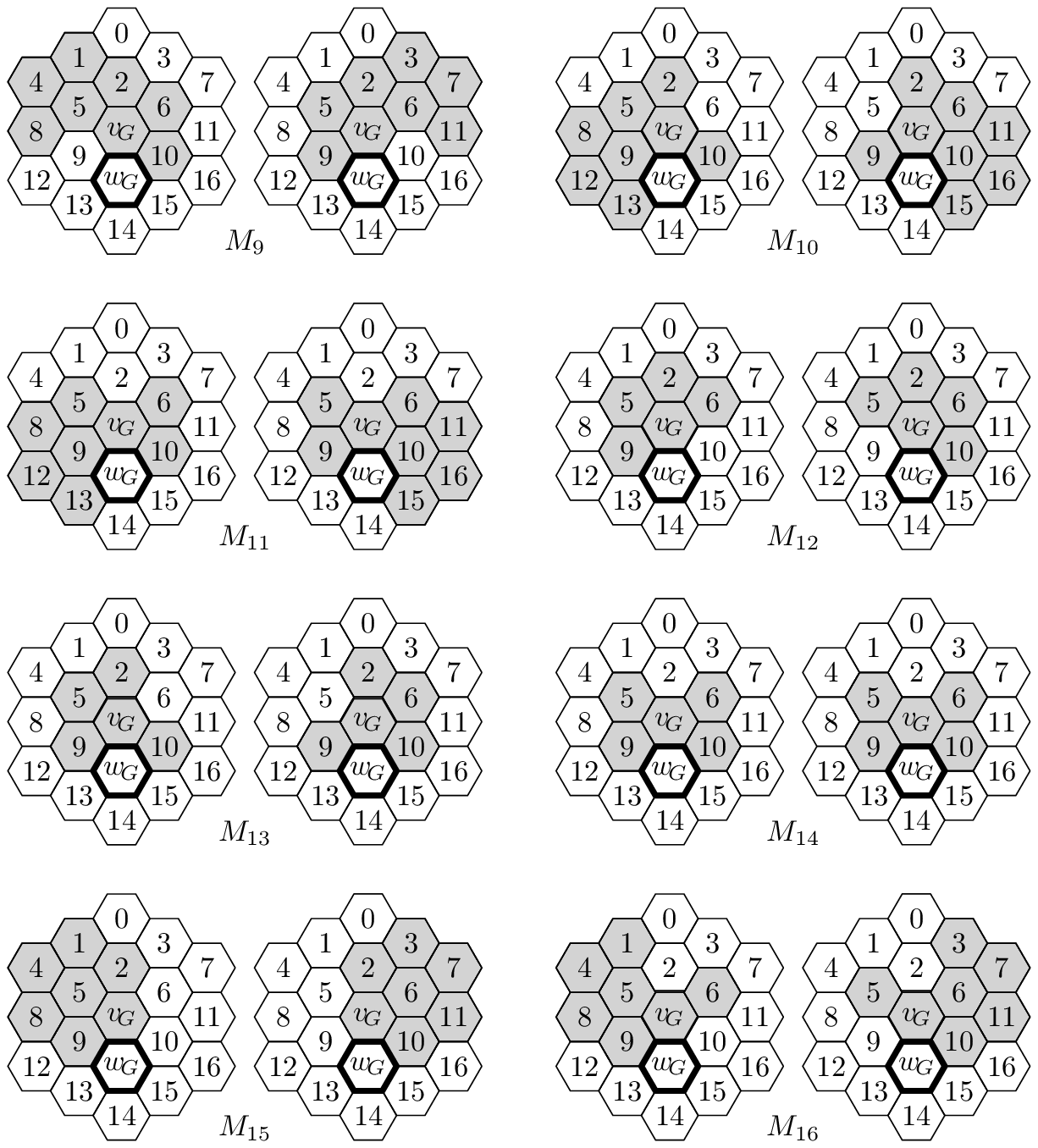}
    \bigskip

    Part 2 of 5
    \caption{The regions $M_1,\dots,M_{39}$ in labelled region~$G$.}
\end{figure}
\begin{figure}[p]
    \addtocounter{figure}{-1}
    \centering
    \showgraphics{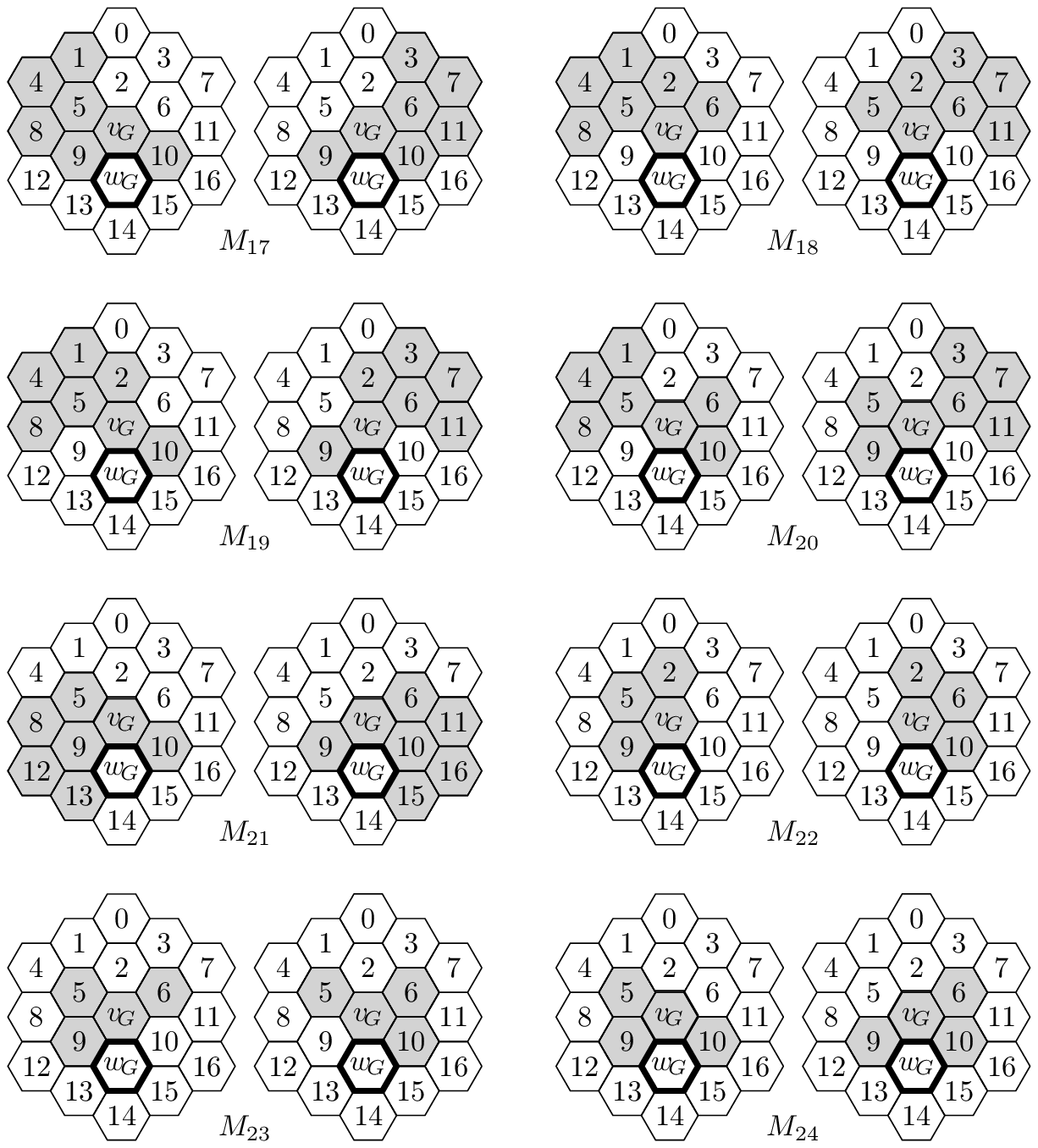}
    \bigskip

    Part 3 of 5
    \caption{The regions $M_1,\dots,M_{39}$ in labelled region~$G$.}
\end{figure}
\begin{figure}[p]
    \addtocounter{figure}{-1}
    \centering
    \showgraphics{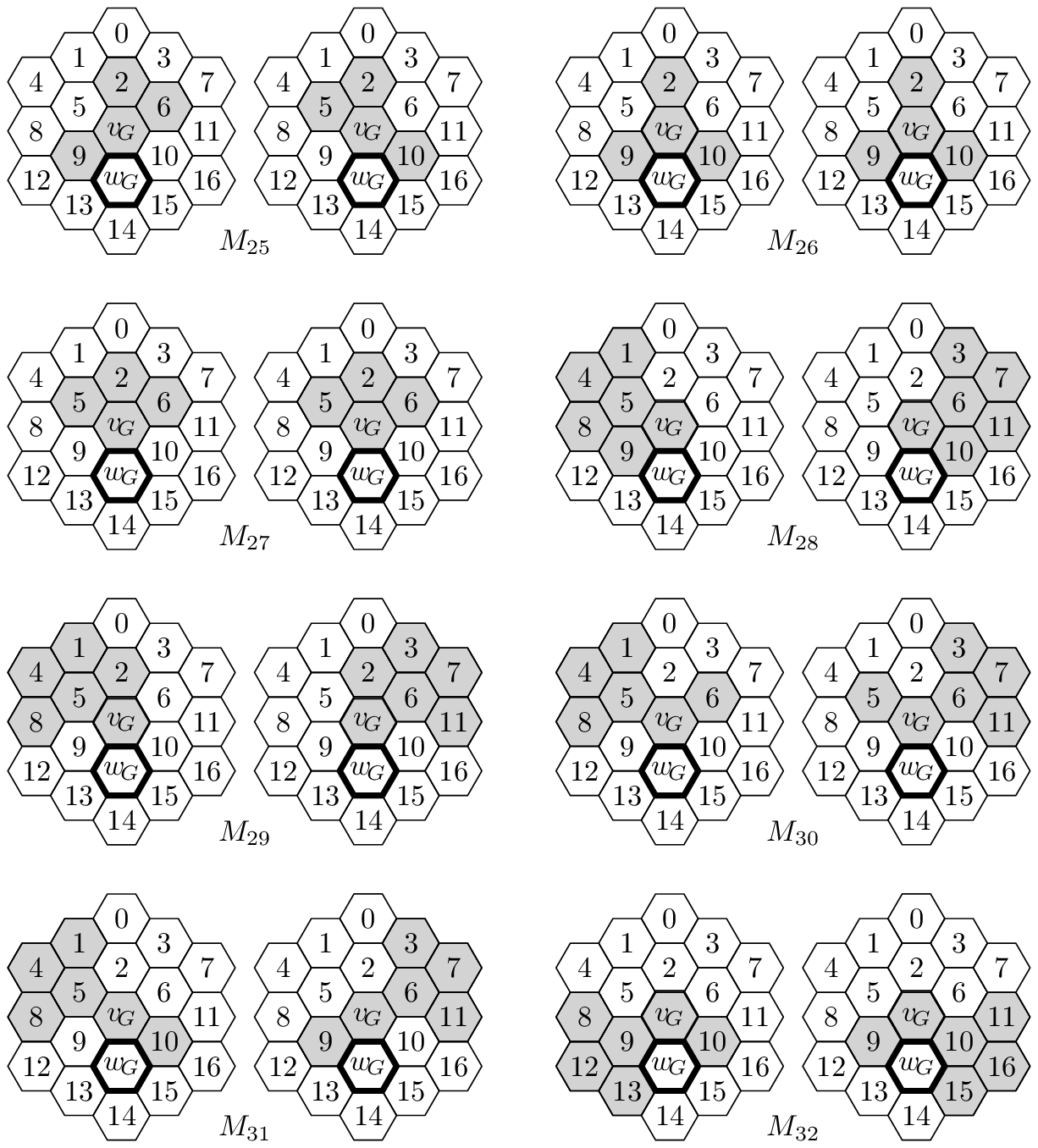}
    \bigskip

    Part 4 of 5
    \caption{The regions $M_1,\dots,M_{39}$ in labelled region~$G$.}
\end{figure}
\begin{figure}[p]
    \addtocounter{figure}{-1}
    \centering
    \showgraphics{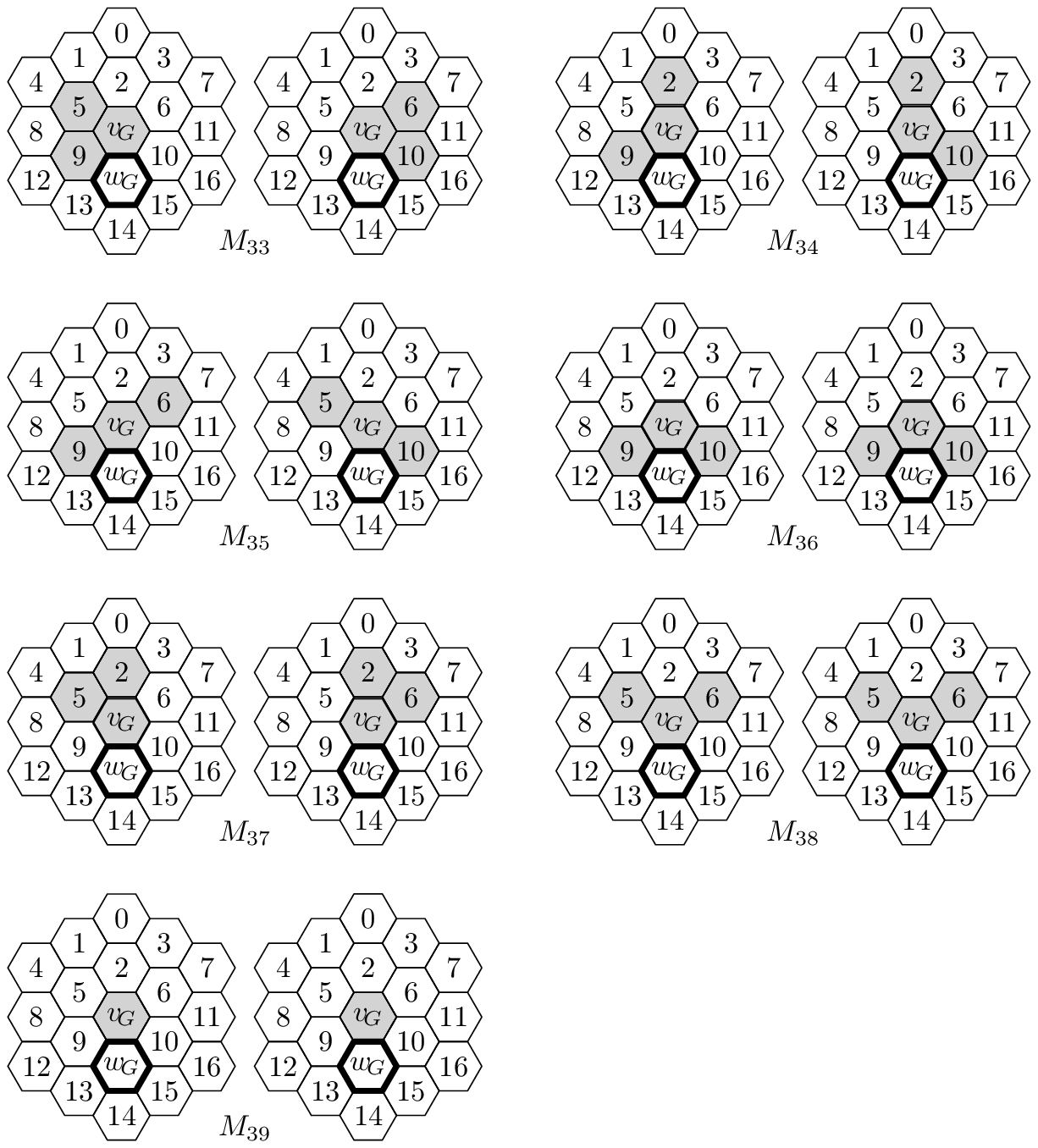}
    \bigskip

    Part 5 of 5
    \caption{The regions $M_1,\dots,M_{39}$ in labelled region~$G$.}
\end{figure}

\end{document}



